\documentclass[a4paper,preprintnumbers,floatfix,superscriptaddress,prx,twocolumn,showpacs,notitlepage,longbibliography]{revtex4-2} 

\usepackage[utf8]{inputenc}  
\usepackage[T1]{fontenc}     
\usepackage[english]{babel}  

\usepackage{physics}
\usepackage{amsmath}
\usepackage{amsfonts}

\usepackage{graphicx} 
\usepackage[colorlinks,citecolor=myBlue,linkcolor=magenta,urlcolor=myBlue]{hyperref}  
\usepackage{amssymb,amsthm,mathtools,dsfont}
\usepackage[capitalise]{cleveref}
\usepackage{physics}
\usepackage{xcolor}
\usepackage{tabularx}
\usepackage[caption=false]{subfig}

\usepackage{orcidlink}
\usepackage{nicematrix} 
\usepackage[shortlabels]{enumitem}
\usepackage{booktabs} 

\newtheorem{Definition}{Definition}
\newtheorem{Lemma}[Definition]{Lemma}
\newtheorem{Corollary}[Definition]{Corollary}
\newtheorem{Conjecture}[Definition]{Conjecture}
\newtheorem{Observation}[Definition]{Observation}

\crefname{table}{Tbl.}{Tbls.}
\crefname{appendix}{App.}{Apps.}
\crefname{section}{Sec.}{Secs.}
\crefname{Definition}{Def.}{Defs.}
\crefname{Lemma}{Lem.}{Lems.}
\crefname{Corollary}{Cor.}{Cors.}
\crefname{Conjecture}{Conj.}{Conjs.}

\DeclareMathOperator{\1}{\mathds{1}}

\DeclareMathOperator{\CC}{\mathcal{C}}
\DeclareMathOperator{\DD}{\mathcal{D}}
\DeclareMathOperator{\OO}{\mathcal{O}}
\DeclareMathOperator{\SSS}{\mathcal{S}}

\definecolor{myOrange}{RGB}{255,153,0} 
\definecolor{myBlue}{RGB}{3,3,133} 
\definecolor{myGreen}{RGB}{64 , 149 , 39} 
\definecolor{myPurple}{RGB}{97,21,67} 
\definecolor{myRed}{RGB}{175 , 24 , 33 } 
\definecolor{myYellow}{RGB}{250 , 183 , 0 } 
\definecolor{myPetrol}{RGB}{23,136, 168}

\newcommand{\node}[1]{\ensuremath{{#1}}}

\newcommand{\neighborhood}[1]{\ensuremath{\mathcal{N}_{#1}}}
\newcommand{\NN}[1]{\neighborhood{#1}}

\newcommand{\NNh}[1]{\ensuremath{\hat{\mathcal{N}}_{#1}}}

\newcommand{\stabdim}[1]{\ensuremath{d_{#1}}}

\newcommand{\fixeddimensionsset}[3]{\ensuremath{L_{#3}^{#1,#2}}}
\newcommand{\ranklist}[2]{\ensuremath{l_{#2}^{#1}}}

\newcommand{\ranktensor}[2]{\ensuremath{T_{#2}^{#1}}}
\newcommand{\tensoreig}[2]{\ensuremath{t_{#2}^{#1}}}

\DeclareMathOperator{\LC}{\text{LC}}

\DeclareMathOperator{\Id}{\1}

\makeatletter
\def\blfootnote{\xdef\@thefnmark{}\@footnotetext}
\makeatother

\begin{document}

\title{Distinguishing Graph States by the Properties of Their Marginals}

\author{\orcidlink{0000-0003-2753-6027}~Lina Vandré{{\large\color{blue}$^*$}}}
\affiliation{Naturwissenschaftlich-Technische Fakult\"at, Universit\"at Siegen, Walter-Flex-Stra\ss e 3, 57068 Siegen, Germany}

\author{Jarn de Jong{{\large\color{blue}$^*$}}}
\affiliation{Electrical Engineering and Computer Science Department, Technische Universit{\"a}t Berlin, 10587 Berlin, Germany}

\author{\orcidlink{0000-0002-9349-4075}~Frederik Hahn}
\affiliation{Electrical Engineering and Computer Science Department, Technische Universit{\"a}t Berlin, 10587 Berlin, Germany}

\author{\orcidlink{0000-0003-0418-257X}~Adam Burchardt}
\affiliation{QuSoft, CWI and University of Amsterdam, Science Park 123, 1098 XG Amsterdam, the Netherlands}

\author{Otfried G\"uhne}
\affiliation{Naturwissenschaftlich-Technische Fakult\"at, Universit\"at Siegen, Walter-Flex-Stra\ss e 3, 57068 Siegen, Germany}

\author{Anna Pappa}
\affiliation{Electrical Engineering and Computer Science Department, Technische Universit{\"a}t Berlin, 10587 Berlin, Germany}

\date{\today}

\begin{abstract}
    \noindent Graph states are a class of multi-partite entangled quantum states that are ubiquitous in quantum information. We study equivalence relations between graph states under local unitaries (LU) to obtain distinguishing methods both in local and in networked settings.
    Based on the marginal structure of graph states, we introduce a family of easy-to-compute LU-invariants.
    We show that these invariants uniquely identify the entanglement classes of every graph state up to 8 qubits and discuss their reliability for larger numbers of qubits. 
    To handle larger graphs, we generalize tools to test for local Clifford (LC) equivalence of graph states that work by condensing large graphs into smaller graphs. 
    In turn, we show that statements on the equivalence of these smaller graphs (which are easier to compute) can be used to infer statements on the equivalence of the original, larger graphs.
    We analyze LU-equivalence in two key settings - with and without allowing for the permutation of qubits. 
    We identify entanglement classes, whose marginal structure does not allow us to distinguish them.
    As a result, we increase the bound on the number of qubits where the LU-LC conjecture holds from 8 to 10 qubits in the setting where qubit permutations are allowed.
\end{abstract}

\maketitle

\blfootnote{{{\large\color{blue}$^*$}}These authors contributed equally to this work.}

\section{Introduction}

Multipartite entanglement plays an important role in many applications of quantum 
information; it is used for example to generate secret keys between multiple parties 
in quantum networks~\cite{Murta_2020_quantum_conference} and as a resource for 
measurement based quantum computing \cite{Raussendorf2003measurement} and error 
correction \cite{Terhal2015errorcorrection} schemes. However, characterizing 
entanglement between more than three parties is not at all
straightforward~\cite{Horodecki2009entanglement,Guehne2009entanglement}. A fundamental 
challenge is that the size of the density matrix increases exponentially with 
the number of parties, making it difficult to simulate large entangled states 
on classical computers. 

Graph states~\cite{hein2006entanglement} are a subset of multipartite entangled 
states with useful properties. They have a convenient graphical representation, 
and, even though they are not as computationally hard to classically simulate as 
general entangled states, they are used in many of the above-mentioned applications. 
Graph states can be treated with respect to associated graphs; the association between graphs and graph states (\cref{fig:example_3marg}) will be analyzed later.

\begin{figure}
    \includegraphics[width = 0.95\linewidth]{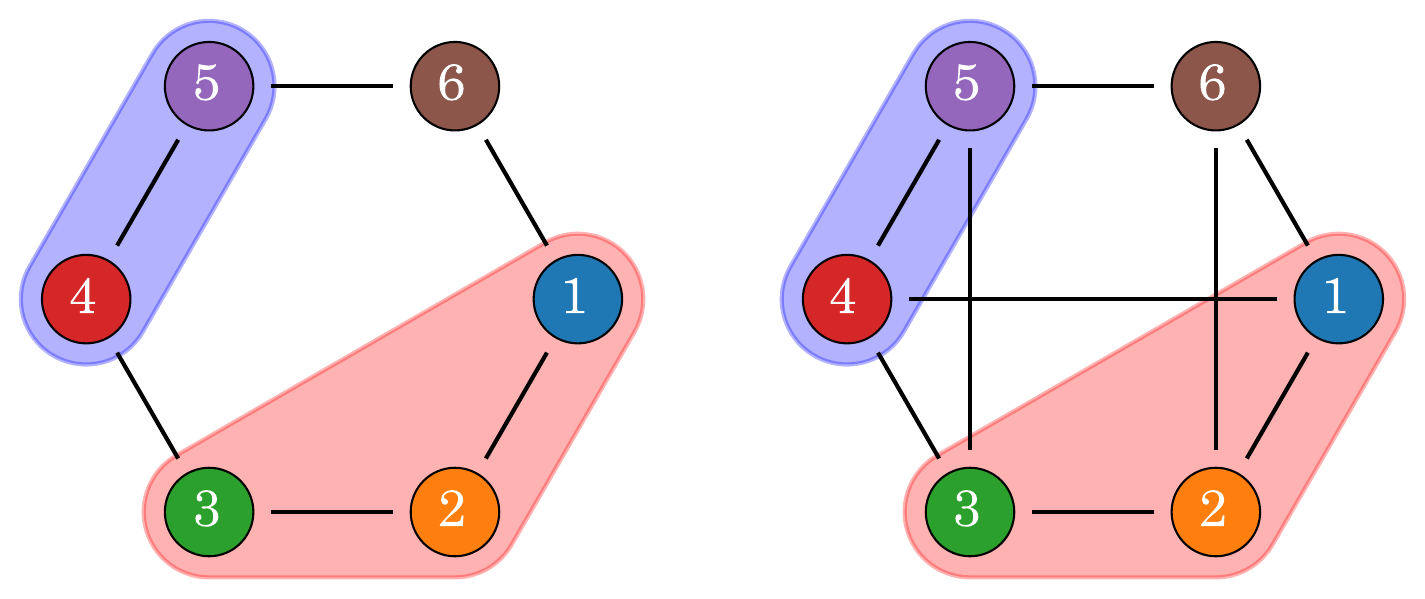}
    \caption{
    \emph{Graph states} are a certain kind of multipartite entangled states where qubits are represented by nodes of a graph, and the graphs' edges represent entangling gates.
    We are interested in finding out which graphs correspond to states that share the same entanglement properties, or more precise, if they are equivalent under local unitary operations. In our methods of determining equivalence an integer number $d$ can be assigned to every subset of nodes, which is necessarily the same if the corresponding states have the same entanglement properties. The subsets highlighted in purple have the same number $\stabdim{\{4,5\}}$ for both graphs in the figure, so no conclusion can be drawn yet. However the two graphs have different number $\stabdim{\{1,2,3\}}$ for the subsets highlighted in red. Thus, the two associated graph states cannot be equivalent under local unitary operations.
    }
    \label{fig:example_3marg}
\end{figure}

In general, the protocols that use entangled states implement specific operations 
that can be classified as either local or global. Local operations affect single 
qubits and are easier to perform. In contrast, global operations target multiple 
qubits. They are more involved and especially require them to be in close proximity to be feasible. In scenarios such 
as key generation, where participants are separated by distance, one is typically 
limited to local operations. 
In such remote settings, it is relevant to determine whether a given entangled state can be transformed into a different state by 
applying local operations, possibly assisted by classical communication, 
a setting commonly known as \emph{LOCC}, or \emph{SLOCC} when stochastic transformations are considered that succeed with a non-unit probability.
In a first approximation, one can consider merely the set of local
unitary (LU) operations, which do not make use of communication between the parties, neither quantum nor classical.
In scenarios such as quantum computing, which are carried out at one location, the set of operations may still be limited. Single qubit gates in general have a much better fidelity than multi qubit gates \cite{Majidy_Wilson_Laflamme_2024}. Furthermore, unitary operations on single qubits preserve the entanglement properties. In contrast to network scenarios, the permutations of qubits may not be of interest. In such scenarios, determining whether a given entangled state can be transformed into a different state up to qubit permutations by applying LU operations is of relevance. 

For checking equivalence of two pure multi-qubit states under LU operations a 
constructive algorithm exists, where LU-equivalence can be decided by solving a finite set of  equations \cite{Kraus2010_PRL}. A minimal number of polynomial equations  needed to decide LU-equivalence of $n$-qubit states is discussed in \cite{Maciazek2013howmany}.

Perhaps surprisingly, the equivalences of graph states under SLOCC and LU operations coincide \cite{hein2006entanglement}, so that only LU operations need to be considered in determining the equivalence of two graph states. Moreover, in many cases it is interesting to look at a specific subset of LU operations, known as local Clifford (LC) operations. These are the subset of LU transformations that map the set of multi-qubit Pauli eigenstates to itself and, more importantly, have a clear graphical representation \cite{Van_den_Nest_2004_graphical}, providing a powerful tool in the study of equivalence of graph states.

By applying LU or LC operations, 
a given graph state can be transformed into several other graph states. The collection of all the graph states that can be reached from a given starting state, with or without the permutation of qubits, is called the 
LU/LC-orbit of the starting state. Graph states in the same LU/LC-orbit are 
called \emph{LU/LC-equivalent}. It was once conjectured that LU- and LC-orbits with permutations for 
graph states coincide \cite{nestLocalUnitaryLocal2005}, but this conjecture
has been refuted by discovering counterexamples with 27 and more qubits~\cite{ji2010lulcconj,Tsimakuridze2017Graphstates}.

Nevertheless, for states 
with 8 or fewer qubits, LU- and LC-orbits with permutations do coincide~\cite{cabelloEntanglementEightqubitGraph2009}.
We increase this bound to 10. Further, we consider LU and LC operations without qubit permutations.
Here, we demonstrate that the LU and LC orbits for up to 8 qubits coincide.

Determining the orbit of a given graph state is a hard problem, even when 
restricting to just LC operations: even counting the number of graph states in an LC-orbit is 
shown to be \#{P}-complete \cite{Dahlberg2020Counting} and therefore at least as hard as an NP-complete problem. Still, an efficient polynomial algorithm capable of determining LC-equivalence between two input graphs was discovered by Bouchet~\cite{bouchet1993Recognizing} and popularized in the physics community in Ref.~\cite{vandennestEfficientAlgorithmRecognize2004}.

Despite the mentioned  general methods for demonstrating the LU-equivalence of generic states \cite{Kraus2010_PRL}, devising an efficient algorithm for this task remains an unresolved challenge. Earlier works \cite{Hein2004_7qubits,Hajdusek_2013} explored the connection between multipartite entanglement of graph states and specific graph-theoretical properties. Following these investigations, studies have been performed \cite{PhysRevA.102.022413,PhysRevA.106.062424} on how to distinguish some non-LU-equivalent states based on the entanglement properties of higher order marginals, while a recent paper~\cite{burchardt2023foliage}  provided polynomial invariants to negatively answer questions about both LC-and LU-equivalence in certain cases based on two-body marginal properties.

In this manuscript, we present a systematic approach to characterize LU-equivalence of graph states based on properties of their marginals. {We consider both cases, LU/LC-equivalence with and without qubit permutations. Our methods contribute to a better visual understanding of entanglement properties in graph states. Furthermore, they help to detect and exclude potential counterexamples of the LU-LC conjecture. Our methods scale polynomially in the number of nodes and linearly in the number of graphs we wish to compare.

The article is structured as follows. In \cref{sec:prelim} we introduce graphs, graph states, LU- and LC-equivalence as well as marginal states. In \cref{sec:tools} we explain tools that can be used to decide LU-equivalence and show that the LU-LC conjecture  is also true for graph states with up to 8 qubits without permutation. 
In \cref{sec:examples}  we discuss how to detect LU-inequivalence when our methods cannot be applied and increase the bound of the LU-LC conjecture with permutations to 10.
We continue by introducing graph simplification tools which are invariant under LC transformations in \cref{sec:lctools}.  In \cref{sec:complex} we discuss the computational equivalence of our tools and we conclude our work with a summary in \cref{sec:sum}.

\section{Preliminaries} \label{sec:prelim}

\subsection{Graphs and Metagraphs} 
\begin{Definition}
    A graph $G = (V,E)$ is a set of nodes $V$  and 
edges $E \subseteq V \times V$ connecting two nodes.
\end{Definition}
In this paper we only consider simple connected graphs. Graphs are connected if for any two nodes $i,j \in V$ we can find a path, i.e., a sequence of adjacent edges, that connects them. Graphs are simple if they contain neither self-loops of nodes nor multiple edges between two different nodes. Finally, we note that our methods easily generalise to disconnected graphs. 

Graphs can also be represented by their adjacency matrices; for a graph with $n = \vert V \vert$ nodes, this would be the symmetric $n \times n$ matrix $\Gamma_{G}$ that encodes the edges as
\begin{align}
\Gamma_{G}(i,j) = \begin{cases}
    1& \text{ if } (i,j) \in E, \\
    0& \text{ otherwise}.
\end{cases}
\end{align}
When context permits, the subscript will be omitted, and just $\Gamma$ will be written.
The $i$-th column $\eta_{i}$ of $\Gamma_{G}$ indicates node $i$'s neighbors, with $\eta_{i}(j) = 1$ if node $j$ is a neighbor, and $0$ otherwise. The \textit{neighborhood} \neighborhood{i} of a node $\node{i}$ is the collection of all other nodes of the graph that share an edge with it:
$$
\neighborhood{i} \coloneqq \{\node{j} \in V \mid (\node{i},\node{j}) \in E\}.
$$
In the same way, we define the neighborhood of a set $M\subsetneq V$ as the set of nodes adjacent to at least one node in $M$:
\begin{align}
    \NN{M} \coloneqq \lbrace v \in V \setminus M \mid \exists m \in M : (v,m) \in E   \rbrace.
\end{align}

In some cases, we are interested in the structure of a set $M \subset V$ of nodes 
and its neighborhood, but not in the entire graph $G$. For that purpose, we define a 
\emph{metagraph} $G_M$,
which depends on a given graph $G = (V,E)$ and a subset of its nodes, $M \subset V$. 
Such a metagraph consists of two types of nodes. Type-1 nodes are simply the nodes in $M$. 
For every non-empty subset $M'\subseteq M$, there exists a type-2 node labeled by 
that subset between square brackets, i.e.,\ $[M']$. 
The rules for connecting the nodes of the metagraph are as follows: Type-1 
nodes are connected to each other if and only if they are connected in the original 
graph $G$. A type-2 node with label $[M']$ is connected to all the type-1 nodes 
in $M$ that appear in their label $[M']$ if and only if the following condition
is met: There exists at least one node $v$ in the original graph $G$, but outside
of $M$ such that 
the intersection of its neighborhood with $M$ equals $M'$.  If there is no $v$ with this property, 
the type-2 node $[M']$ remains isolated.

Note that there are never edges between two type-2 nodes in a metagraph. 
Finally, we are also ignorant as to how many nodes in $V \setminus M$ of the original graph are
adjacent to nodes in $M$.
 We formalize this definition as follows, where we denote with $\mathcal{P}(M)$ the power-set of $M$, i.e. the set $\mathcal{P}(M)=\{N: N\subseteq M \}$ of all subsets of $M$. 

\begin{Definition}[Metagraph] 
\label{def:metagraph}
    For a graph $G = (V,E)$ and a set $M \subsetneq V$, the metagraph of $M$ is defined 
    as the graph $G_M = (M \cup [\mathcal{P}(M) \setminus \emptyset], E_M)$, where each 
    node $v$ of $G_M$ is of one of two types: either $v\in M$ or $v\in \mathcal{P}(M) 
    \setminus \emptyset$. The edge set $E_M$ contains all edges of the induced subgraph 
    on $M$. 
    Furthermore, if in the initial graph $G$ there exists a node $v'\in V\setminus M$ such that $\mathcal{N}_{v'} \cap M=w$, the node $w\in \mathcal{P}(M) \setminus \emptyset$ is connected to all nodes $v\in w$. Otherwise, the node $w \in \mathcal{P}(M) \setminus \emptyset$ remains isolated.
\end{Definition}

\begin{figure}[t]
    \includegraphics[width = 0.45\textwidth]{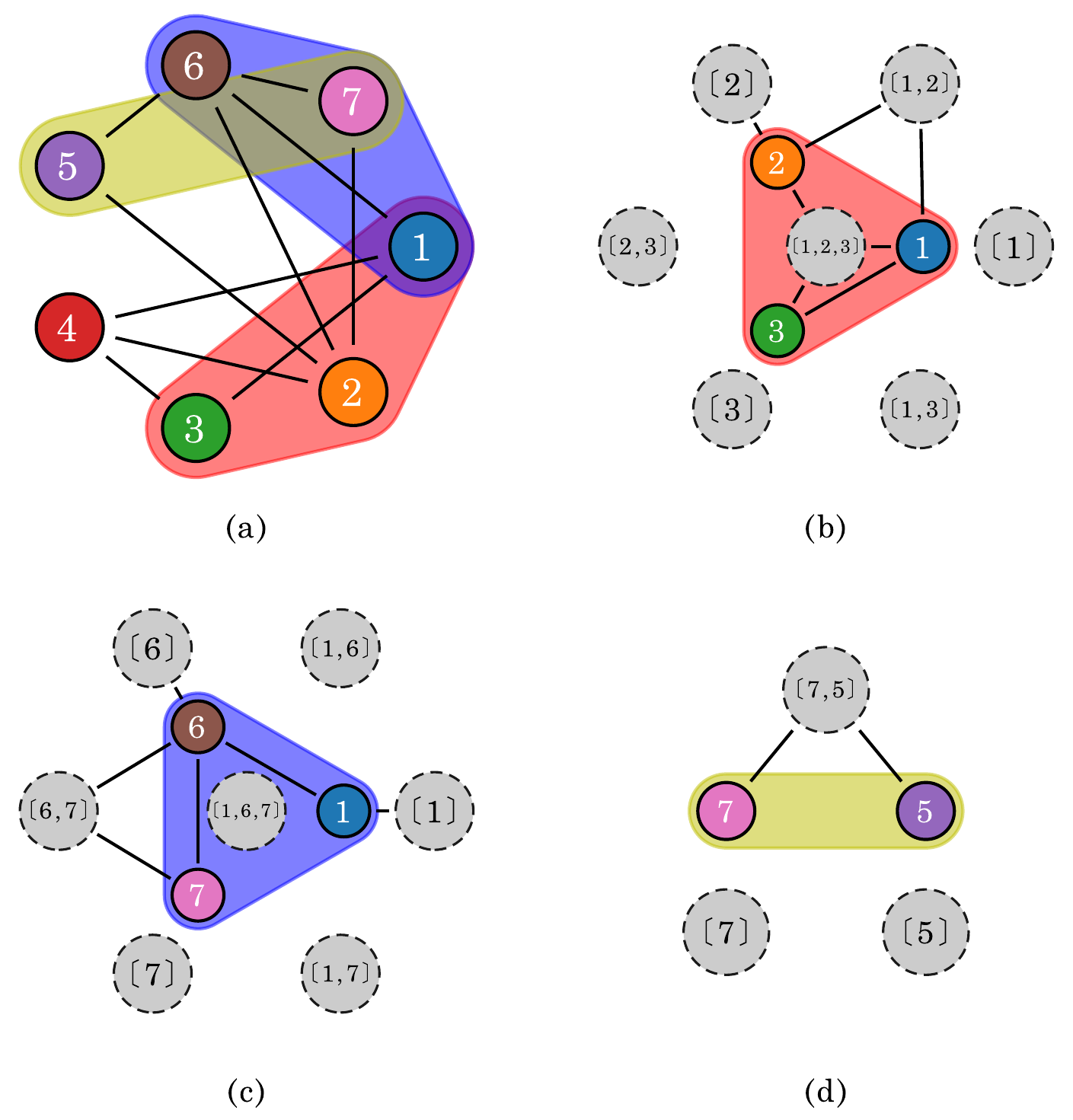}
    \caption{\textbf{$(a)$} A graph $G = (V,E)$ with $3$ different highlighted sets $M$, of size $2$ and $3$. \textbf{$(b),(c),(d)$} The  metagraphs that rise from the different highlighted sets.
    }
    \label{fig:example_metagraphs}
\end{figure}
Let us consider the explicit example in \cref{fig:example_metagraphs}, and specifically the upper right (red) metagraph (b). The metagraph $G_{\lbrace 1,2,3 \rbrace}$ has type-1 nodes $M = \lbrace 1,2,3 \rbrace$, as well as type-2 nodes $\mathcal{P}(M) \setminus \emptyset = \lbrace \lbrace 1 \rbrace, \lbrace 1,2 \rbrace, \lbrace 1,2,3 \rbrace, \dots \rbrace$. Between nodes in $M$, the metagraph has the same edges as $G$. Edges between nodes in $M$ and  $\{ \{ 2 \}, \{ 1,2 \}, \{1,2,3  \}  \}$ can be explained as follows:
There is an edge between node $2 \in M$ and node $\{ 2 \} \in \mathcal{P}(M)$, since for  node $5 \in V$ of the original graph, we have $\mathcal{N}_{5} \cap M=\{ 2 \}$. We further find edges between nodes $1, 2 \in M$ and node $\{1, 2 \} \in \mathcal{P}(M)$, since for  node $6 \in V$ of the original graph, we have $\mathcal{N}_{6} \cap M=\{1, 2 \}$. Since $\mathcal{N}_{4} \cap M=\{1, 2, 3 \}$, we also have edges between nodes $1, 2, 3 \in M$ and node $\{1, 2, 3 \} \in \mathcal{P}(M)$.

\subsection{Graph States}

Graph states are quantum states composed of multiple qubits and related to a graph $G = (V,E)$, where nodes and edges represent qubits and entangling gates, respectively.
\begin{Definition}
   The graph state corresponding to a graph $G = (V,E)$ is defined as
\begin{align}
    \ket{G} \coloneqq \prod_{e \in E} \text{CZ}_e \ket{+}^{\otimes \lvert V \rvert }, 
    \label{eq:grstate}
\end{align}
where $\text{CZ}_e$ is a CZ gate, acting on qubits in the 
edge $e$.
\end{Definition}

There is an equivalent definition using stabilizer operators. For 
each node $i \in V$ we can define a generator, 
\begin{align}\label{eq:grstategen}
g_i \coloneqq X_i \bigotimes_{j \in \neighborhood{i}} Z_j,
\end{align}
where $X_i$, $Y_i$, and $Z_i$ are the Pauli matrices acting on the $i$-th qubit and $\neighborhood{i}$ is the neighborhood of $i$. The products of the generators $g_i$ form the \textit{stabilizer}  $\SSS(G)$, a commutative group with $2^n $ elements. The graph state $\ket{G}$ is then defined as the unique $n$-qubit $+1$ eigenstate of all generators $g_i$ or simply by 
\begin{align}
    \dyad{G} = \frac{1}{2^n} \sum_{s \in \SSS(G)} s. \label{eq:graph_stabilizer}
\end{align}

\begin{figure}[t]
    \includegraphics[width = 0.45\textwidth]{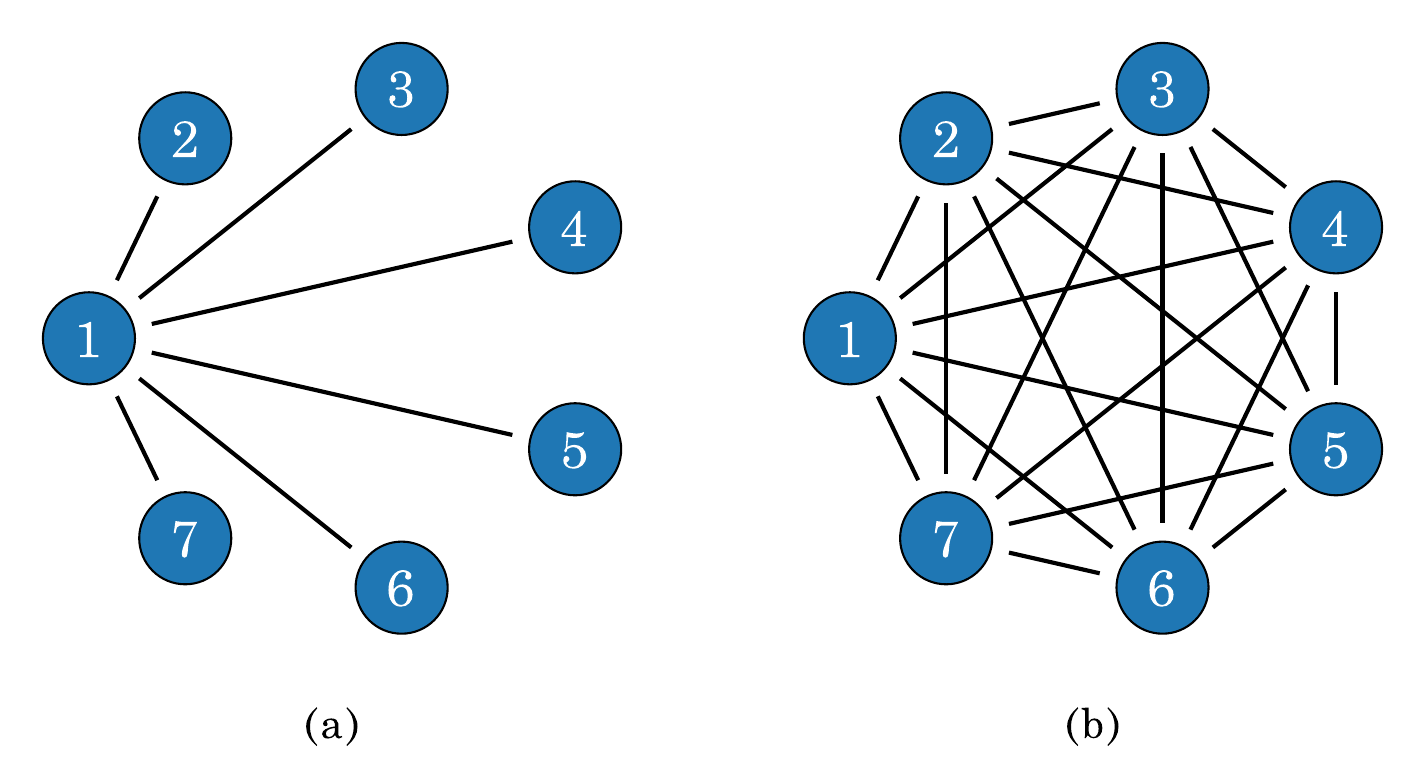}
    \caption{\textbf{$(a)$} The \textit{star} graph, whose corresponding graph state is equivalent to the $\ket{\mathrm{GHZ}}$ state. \textbf{$(b)$} Local operations on the qubits of the $\mathrm{GHZ}$ state result in a different (albeit locally equivalent) graph state, which is represented by the complete graph.
    }
    \label{fig:example_starvscomplete}
\end{figure}

A prominent example of a graph state is the Greenberger-Horne-Zeilinger (GHZ) state 
\cite{GHZ}, used in quantum communication protocols to establish secret keys \cite{proietti2021experimental,epping2017multi}, share quantum secrets \cite{Hillery1999Qsecret}, and communicate anonymously \cite{ACKA,expACKA}. An $n$-qubit GHZ state can be represented by a fully connected graph or a star graph of $n$ nodes, as shown in \cref{fig:example_starvscomplete}.
Another noteworthy example of graph states are cluster states; these are  represented by grid-like graphs and play an important role in measurement-based quantum computing \cite{Raussendorf2001,Nielsen_2006_cluster_QC}.

In this paper we consider labeled and unlabeled graphs. 
We use unlabeled graphs when all permutations of the vertices are considered equivalent.
Labeled graphs are used when we are interested in a specific assignment of qubit indices  where index permutations in general result in a different state.

\subsection{Equivalence under Local  Operations}

In this section, we first define the problem of LU-equivalence. We then discuss a well-known subset of local unitaries, the local Clifford (LC) unitaries, and the equivalence of graph states under such LC operations. Finally, we emphasize the distinction between LU/LC-equivalence of graph states with and without allowing for permutation of qubits, corresponding to unlabeled and labeled graphs, respectively.

Graph states naturally arise in network scenarios, where each node in the network represents a single qubit. Since the nodes are usually spatially distant from each other, studying the effect of \textit{local} operations on graph states is of practical interest.

\subsubsection{LU-Equivalence}

Local operations can transform one graph state into another, which leads to the question whether two graph states are equal up to local operations.
More specifically, two graph states $\ket{G_{1}}$ and $\ket{G_{2}}$ are \textit{local unitary} or LU-equivalent if and only if $\ket{G_{1}} = U\ket{G_{2}}$ and the $n$-qubit unitary operator $U \in \mathcal{U}(2^{n})$ allows a decomposition into a tensor product of single-qubit operations
\begin{equation}
    U = \bigotimes_{i=1}^{n}U_{i} \quad \mid U_{i} \in U(2).
\end{equation}

For example, the graph states corresponding to the star graph and the fully connected graph are local unitary equivalent to the GHZ state in its standard form. For $n$-qubits this would be:
\begin{align}
    \ket{\mathrm{GHZ}_n} = \frac{1}{\sqrt{2}} \left( \ket{0_1 0_2 \dots 0_n} + \ket{1_11_2 \dots 1_n}  \right).
\end{align}
For the $n$-node star graph with the first node as central node (as shown in \cref{fig:example_starvscomplete}), one finds:
\begin{align}
    \ket{G_{\text{star},n}} 
    &= \frac{1}{\sqrt{2}} \left( \ket{0_1 +_2 \dots +_n} + \ket{1_1 -_2 \dots -_n} \right). 
\end{align}
Applying local, unitary Hadamard gates on all qubits except the first, we get 
$\ket{\mathrm{GHZ}_n}$. In a similar way, the equivalence with the fully connected graph can also be shown.

In general, it is hard to decide whether two graph states are LU-equivalent. There exists a constructive algorithm, where LU-equivalence of general states can be decided by solving a finite set of  equations \cite{Kraus2010_PRL}. However, the number of polynomials grows exponentially with the number of qubits $n$ \cite{Maciazek2013howmany}. It is therefore of interest to find efficient tools that can decide LU-equivalence. In this paper we discuss a necessary but not sufficient method to determine LU-equivalence of graph states.

\subsubsection{LC-Equivalence and Local Complementation} \label{sec:LCequivalence}

From the local unitaries we can turn our attention to one of their subsets. The 
\textit{Clifford group} is the normaliser of the $n$-qubit Pauli group $\mathcal{P}_{n}$ 
in the unitary group $\mathcal{U}(2^{n})$:
\begin{equation}
    \mathcal{C}_{n} = \{C \in \mathcal{U}(2^{n}) | C \mathcal{P}_{n} C^{\dagger} = \mathcal{P}_{n}\}.
\end{equation}
In practice, this means that Clifford unitaries map strings of Pauli matrices to strings 
of Pauli matrices.

The Clifford matrices that decompose into single-qubit tensor products form the 
\textit{local Clifford group} $\mathcal{C}_{n}^{l}$. Two graph states $\ket{G_{1}}$ 
and $\ket{G_{2}}$ are \textit{local Clifford} or LC-equivalent if and only if $\ket{G_{1}} = C\ket{G_{2}}$ 
for some $C \in \mathcal{C}_{n}^{l}$. 
The local Clifford group 
has $24$ elements and is generated by the Hadamard gate $H$, the phase gate $S$, and the Pauli matrix $X$. As a group with finite number of elements, it is a proper subgroup of the group of local unitaries.
Remarkably, there is a simple relationship between local Clifford operations on graph states and so-called \textit{local complementations} on their associated graphs \cite{Van_den_Nest_2004_graphical}.

 \begin{Definition}\label{def:local_complementation} 
 For each node $i \in V$ of a graph $G=(V, E)$, we can define a \textit{locally complemented graph} $LC_{i}(G)$ with adjacency matrix
	$
	\Gamma_{LC_{i}(G)} \coloneqq \Gamma_{G}+\Gamma_{{\neighborhood{i}}} \ (\bmod \ 2),
	$
	where 
  \begin{align}
    \Gamma_{{\neighborhood{i}}}(j,k) = \begin{cases}
        1& \text{ if } j,k \in \neighborhood{i}, j \neq k \\
        0& \text{ otherwise}.
    \end{cases}
  \end{align}
	is the adjacency matrix of the complete graph on the neighborhood $\neighborhood{i}$ of $i$ and empty on all other nodes and $\Gamma_{G}$.
\end{Definition}

In other words, local complementation `flips' the adjacency relations between the neighbors of $i$: if two nodes in \neighborhood{i} are connected by an edge in $G$, they are not connected in $\LC_{i}(G)$, and vice versa. Each $\LC_{i}$ operation only affects the edges in the neighborhood of $\node{i}$; all other edges remain unchanged.

There is a one-to-one correspondence between the local complementation orbit of a given graph and the orbit under local Clifford operations of the corresponding graph state \cite{Van_den_Nest_2004_graphical}.
 The explicit local Clifford operation $\LC_{i} \in \mathcal{C}_{n}^{l}$ corresponding to a local complementation is given by 
\begin{equation}
     \LC_i = \sqrt{-i X_i} \bigotimes_{j \in \NN i} \sqrt{i Z_j}.
\end{equation}
{An example of local complementation is shown in \cref{fig:example_starvscomplete}. If we perform local complementation on the central node of the star graph, all other nodes get pairwise connected, which results in the fully connected graph. Local complementation in, for example, node 2 of the star graph does not change the graph. Local complementation on any node $i$ of the fully connected graph leads to a star graph with node $i$ as the center. We previously showed that the star graph state is LU-equivalent to the GHZ state in its common form. Applying a local complementation shows that the fully connected graph state is LC- (and therefore LU-) equivalent to the GHZ state as well.}

Note that the effect of an arbitrary local Clifford on a graph state does not have to result 
in a graph state again. In general, the effect of local Cliffords on graph states will result 
in different stabilizer states. 

Fortunately, every stabilizer state is itself LC-equivalent to some graph state, and such graph states can be found efficiently \cite{Van_den_Nest_2004_graphical}. 
Since this relation is not unique, we say that a stabilizer state is LC-equivalent to a graph state orbit. 

Combining this insight with the efficient algorithm for determining whether two graphs are in each other's local complementation orbit \cite{bouchet1993Recognizing} results in an efficient method for determining whether two stabilizer states are LC-equivalent \cite{vandennestEfficientAlgorithmRecognize2004}.

It turns out that for unlabeled graphs of up to 8 qubits, LU-equivalent graph states are also LC-equivalent \cite{Hein2004_7qubits,cabelloEntanglementEightqubitGraph2009}. In this regime it therefore is sufficient to use the efficient methods of LC-equivalence checking to decide the in principle hard question of LU-equivalence. The hope that the same applies to graph states with an arbitrary number of qubits was formulated as the LU-LC conjecture \cite{nestLocalUnitaryLocal2005}. However, the conjecture was falsified by finding counterexamples of $27$ node graphs \cite{ji2010lulcconj,Tsimakuridze2017Graphstates}. Therefore, different methods for deciding LU- and LC-equivalence of graph states are needed.

\subsubsection{Distinction Between Local Equivalence of Labeled and Unlabeled Graphs
}

As mentioned earlier, we differentiate between labeled and unlabeled graphs.

On the one hand, our definition of LU- and LC-equivalence is from a network point of view, where a labeling of the nodes 
of the graph corresponds to the qubit locations of the respective graph state.

On the other hand, previous work in entanglement theory mostly considered unlabeled graphs \cite{cabelloEntanglementEightqubitGraph2009,nestLocalUnitaryLocal2005, danielsen2005database12qubits}, where it is not relevant whether certain qubits of a graph state are at a fixed location and have a corresponding label.
In other words, comparing unlabeled graph states under local operations amounts to studying so-called entanglement classes. Two graph states belong to the same entanglement class if there exists an arbitrary permutation of the qubits of the first graph state that results in a graph state that is in the (LU- or LC-) orbit of the second graph state.

\begin{figure}
    \centering
    \includegraphics[width=\linewidth]{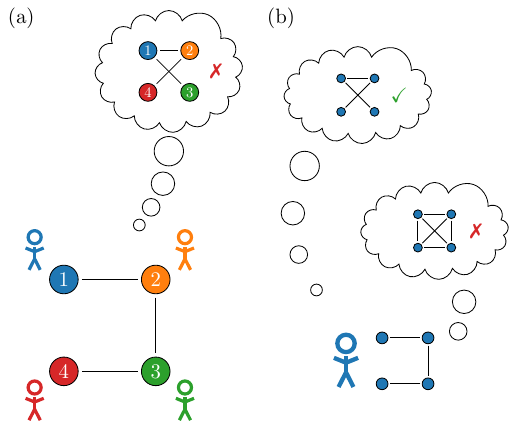}
    \caption{Examples of labeled and unlabeled graphs of four nodes.
    \textbf{$(a)$} In a network scenario, the qubits of a state are distributed among distant parties and particles cannot be exchanged. Due to the labeling, the shown graphs are different and the corresponding graph states cannot be transformed into each other by applying local operations.
    \textbf{$(b)$} If the state is not distributed, it is easy to relabel qubits and it is interesting whether two states share the same entanglement properties. The two graphs which both have 3 edges are equal. It is interesting to know whether the third graph is LU- or LC-equivalent to the others for some labeling. 
    }
    \label{fig:labelledvsunlabelled}
\end{figure}

In this article, we discuss both cases.
An example that illustrates why both cases should be considered separately is shown in \cref{fig:labelledvsunlabelled}.

As another example of the difference between the two cases, consider the 6 unlabeled 4 node connected graphs, and the 38 labeled 4 node connected graphs. They can either be (in the case of unlabeled graphs) considered as two separate entanglement classes, or (in the case of labeled graphs) as four separate LU-orbits. One of these entanglement classes and one of these orbits are those associated with the GHZ state. The other entanglement class and three LU-orbits are associated with the linear cluster state and are detailed in \cref{app:labelvsunlabel}.

\subsection{Marginal States}
In the following section we discuss how  structural properties of graphs are related to properties  of reduced density matrices (i.e.\ marginals) of graph states,

\begin{Definition}
    The marginal state $\rho_M$ on a proper subset of nodes $M \subsetneq V$ is defined as
    \begin{align}
    \label{eq:marginal}
        \rho_{M}(G) \coloneqq \Tr_{V \setminus M}(\dyad{G}) = \frac{1}{2^n} \sum_{S \in \SSS(G)} \Tr_{V \setminus M}(S),
    \end{align}
where $\Tr_{V \setminus M}$ is the trace over all systems in $V \setminus M$. 
\end{Definition}

Since Pauli matrices are traceless, for sufficiently small $M$ we have $\Tr_{V \setminus M}(S) = 0$ in many cases. There are two special cases where the partial trace of a stabilizer element is not zero. Firstly, for the trivial stabilizer element $\1_V \in \SSS(G)$ we have $\Tr_{V \setminus M}(\1_V) = 2^{\vert V \setminus M \vert} \1_M$, where the index of $\1$ indicates the size of the identity matrix. Secondly, depending on the choice of $M$, there may be stabilizer elements whose non-trivial supports are contained in $M$, i.e.~all indices of Pauli matrices are contained in $M$. In this case, the partial trace of the stabilizer element is also different from zero.

The \textit{reduced stabilizer} is the subgroup $\SSS_{M} \subseteq \SSS$ defined
\begin{align}
    \SSS_M &\coloneqq \lbrace S \in \SSS(G) : \Tr_{V \setminus M}(S) \neq 0 \rbrace \label{eq:stab_set_tr} \\
    &= \lbrace S \in \SSS(G) : \mathrm{supp}(S) \subset M \rbrace, \label{eq:stab_set_supp}
\end{align}
where $\mathrm{supp}(S)$, the \textit{support} of a Pauli operator $S$, is the collection of tensor-subspaces on which it acts non-trivially. 
Directly from the definition, it follows that $\SSS_M$ forms a group. Every element of $\SSS_{M}$ squares to the identity and commutes with any other element, and therefore the order of the group $\SSS_M$ is a power of two. Consequently, we can define the following integer number
\begin{equation}
\label{eq:d_M}
\stabdim{M} := \log_2(\vert \SSS_M \vert),
\end{equation}
as the rank of the group $\SSS_{M}$, i.e.,~the minimum number of elements in a generating set of $\SSS_{M}$.

$\SSS_{M}$ is an abelian subgroup of $\SSS$ with at most $2^{\vert M \vert}$ elements. Therefore, when $M$ is a proper subset of $V$,  we have the equation
\begin{equation}
    0 \leqslant \stabdim{M} < \abs{M},
\end{equation}
where the last inequality is strict because we only consider connected graphs. This also allows for a straightforward check of connectedness: $\stabdim{M} = \abs{M}$ if and only if a graph is disconnected over the partition $M:M'$.

Notice that the group $\SSS_{M} $ defines the respective marginal state $\rho_{M}(G)$. Indeed, by \cref{eq:marginal}, we have
\begin{equation}
\label{eq:AdNew}
        \rho_{M}(G)=\frac{1}{2^{n}}
        \sum_{S\in \SSS_{M}} \Tr_{V \setminus M}(S).
\end{equation}
Note that the normalization factor is $2^{n}$ instead of $2^{\abs{M}}$ because $\Tr_{V \setminus M}(S)$ is only an $\abs{M}$-qubit Pauli operator up to a factor $2^{n - \abs{M}}$.

\begin{figure}
    \centering
    \includegraphics[width = 0.3\textwidth]{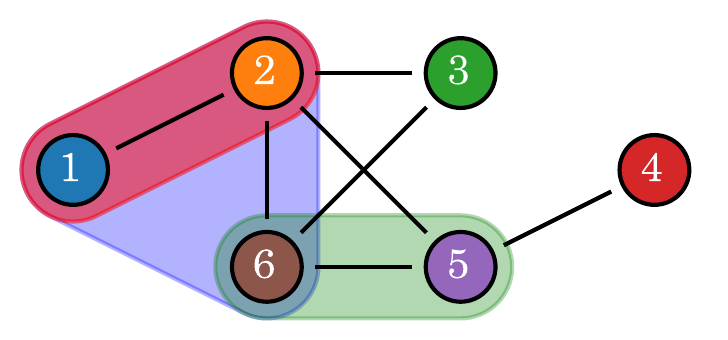} 
    \caption{Example for marginal states on a six node graph. State $\rho_{\lbrace 1,2 \rbrace}$ is stabilized by $\SSS_{\lbrace 1,2 \rbrace} = \lbrace \1, g_1  \rbrace$, while $\rho_{\lbrace 1,2, 6 \rbrace}$ is stabilized by $\SSS_{\lbrace 1,2,6 \rbrace} = \lbrace \1, g_1, g_2 g_6, g_1 g_2 g_6 \rbrace$ and $\rho_{\lbrace 5,6 \rbrace}$ is stabilized by $\SSS_{\lbrace 5,6 \rbrace} = \lbrace \1 \rbrace$. The stabilizer elements are given by $g_1 = X_1 Z_2$, $g_2 g_6 = - Z_1 Y_2 Y_6$, and $g_1 g_2 g_6 =  Y_1 X_2 Y_6$}
    \label{fig:example_marginals}
\end{figure}

\noindent
An example for marginal states is given in \cref{fig:example_marginals}.
Computing one- and two-body marginals of graph states is discussed in details in Ref.~\cite{Gittsovich2010Multiparticle}.

\section{Tools for characterizing LU-equivalence} \label{sec:tools}

We are interested in which graph states can or cannot be transformed into other graph states using local unitary operations. LU-equivalent graph states have in general different reduced states, but certain properties stay invariant under local unitary operations. For example the number of stabilizer elements in $\vert \SSS_M \vert$ or the rank of the reduced state are invariant under LU transformations. In the following we utilize these invariant properties to distinguish between graph states
that are not LU-equivalent. We begin by outlining alternative methods for computing the stabilizer dimension, $d_M$.

The stabilizer dimension can be computed from the rank of the marginal state.
\cref{eq:AdNew} allows us to compute the rank of a marginal state $\rho_{M}$. Indeed, the group structure of the reduced stabilizer implies that $\rho_{M}^{2} = \frac{2^{\stabdim{M}}}{2^{\abs{M}}}\rho_{M}$. It follows that $\rho_{M}$ is a projection operator $\Pi_{M}$ up to a scaling factor of $2^{\stabdim{M} - \abs{M}}$: $\rho_{M} = 2^{\stabdim{M} - \abs{M}} \Pi_{M}$. Note that $\rank(\rho_{M}) = \rank(\Pi_{M}) = \Tr[\Pi_{M}]$, where the last equality holds because $\Pi_{M}$ is a projection operator. Combining these insights, it follows that:
\begin{equation}\label{eq:rankofrhoMcalc}
    1 = \Tr[\rho_{M}] = 2^{\stabdim{M} - \abs{M}}\rank(\rho_{M}).
\end{equation}
This allows us to compute $\rank(\rho_{M})$ in terms of the marginal dimension $\stabdim{M}$:
\begin{align}
   \rank(\rho_M) = 2^{\vert M \vert - \stabdim{M}}.
\end{align}

 The rank of a reduced density matrix is invariant under LU operations, 
 thereby ensuring that the dimension $d_M$ also is invariant under such operations.

Sometimes it will be useful to calculate the marginal dimension \stabdim{M} from the adjacency matrix $\Gamma$. Without loss of generality, we assume that the nodes in $M$ are the first $\abs{M}$ indices of $\Gamma$; the rest is indexed by $M^{\perp} = V \setminus M$. We can then write the adjacency matrix in block form:
\begin{equation}
    \Gamma = \begin{bNiceArray}{c|c}
  \Gamma_{M,M} & \Gamma_{M,M^{\perp}} \\
  \hline
  \Gamma_{M^{\perp},M} & \Gamma_{M^{\perp},M^{\perp}} 
\end{bNiceArray},
\end{equation}
where $\Gamma_{A,B} $ denotes the submatrix of the adjacency matrix $\Gamma$ with rows and columns indexed by subsets $A$ and
$B$ respectively. 
We can calculate the dimension of \stabdim{M} as follows \cite{hein2006entanglement}:
\begin{equation}
    \stabdim{M} = \mathrm{Null}\left(\Gamma_{M^{\perp}, M}\right) = \abs{M} - \text{rank}(\Gamma_{M,M^{\perp}}),
\end{equation}
where $\mathrm{Null}(A)$ denotes the nullity of a matrix $A$, i.e.~the dimension of its kernel, calculated over the binary field. 

In order to determine the dimension $d_M$, we can also use the adjacency matrix $\Gamma^{G_M}$ of the corresponding metagraph $G_M$. Indeed, we have
\begin{align}
    \text{rank}(\Gamma_{M,M^{\perp}})  = \text{rank}\, \Gamma^{G_M}_{M,\mathcal{P}(M)}
\end{align}
where $\Gamma^{G_M}_{M,\mathcal{P}(M)}$ denotes the submatrix of the adjacency matrix $\Gamma^{G_M}$ with rows and columns indexed by nodes in $M$ and columns indexed by subsets of $M$. This can be easily seen, as the matrix $\Gamma^{G_M}_{M,\mathcal{P}(M)}$ is obtained from $\Gamma_{M,M^{\perp}} $ by deleting repetitive columns, see \cref{def:metagraph}.

We conclude this section with outlining the following relation between the marginal dimension \stabdim{M} and the entanglement entropy $E_M (\ket{G})$ with respect to the bipartition $M \vert M^{\perp}$:
\begin{align}
    E_M (\ket{G}) \coloneqq - \Tr(\rho_M \log_2 \rho_M) = \vert M \vert -   \stabdim{M}.
\end{align}

From the fact that $E_M (\ket{G})$ = $E_{M^{\perp}} (\ket{G})$ we get the equality
\begin{equation}\label{eq:d_M=d_m-k}
    \stabdim{M^{\perp}} = \stabdim{M} + \lvert M^{\perp} \rvert - \lvert M \rvert.
\end{equation}

\subsection{Characterizing LU-orbits by Their Rank Invariants} 
\label{sec:rankinvar}
In this section, we discuss how we can use the mentioned invariant measures to distinguish between different LU-orbits of graph states. 
For a given graph $G=(V,E)$, we denote by
\begin{equation}
    \fixeddimensionsset{k}{G}{i} = \{M \subseteq V \mid |M|=k,\, d_M= i\}
\end{equation}
the collection of all subsets $M\subseteq V$ of nodes of size $|M|=k$ such that the corresponding marginal $\rho_M$ is of dimension $d_M=i$. Furthermore, for a given value $k$, we define
\begin{equation}
    \ranklist{G}{k} = \Big(\abs{\fixeddimensionsset{k}{G}{0}}, \dots ,\abs{\fixeddimensionsset{k}{G}{k}}\Big),
\end{equation}
which is a vector of dimensions of sets $\fixeddimensionsset{k}{G}{i}$ for all values $0\leq i\leq k$. Moreover, for any graph $G=(V,E)$ and number $k$, $0\leq k\leq n$, we define the $k$-dimensional tensor 
\begin{equation}\label{def:ranktensor}
    \ranktensor{G}{k} =\big( t_{i_1 \cdots i_k}\big)_{i_1,\ldots,i_k\in V}
\end{equation}
with entries determined by the dimensions of the corresponding marginals, i.e.:
\begin{equation}
    t_{i_1 \cdots i_k} := d_{\{i_1,\ldots,i_k\}}.
\end{equation}
Notice that elements $i_1,\ldots,i_k$ are not necessarily pairwise different, hence the set $M=\{i_1,\ldots,i_k\}$ might be of any size $1\leq |M|\leq k$. Therefore, \ranktensor{G}{k} contains the information about the size of dimensions of all marginal states $\rho_M$ corresponding to the subsets $M\subseteq V$ of size $|M|\leq k$. 
Notice that for any $k$, the tensor \ranktensor{G}{k-1} is embedded into the (generalized) diagonals of \ranktensor{G}{k}, and that \ranktensor{G}{k} is super-symmetric, i.e. $t_{i_1 \cdots i_k}=t_{\sigma (i_1 ) \cdots \sigma (i_k )}$ for any arbitrary permutation $\sigma$. 
Moreover, because of \cref{eq:d_M=d_m-k}, any $\ranktensor{G}{k}$ for $k > \lfloor \frac{n}{2} \rfloor$ can be directly computed from $\ranktensor{G}{n - k}$.

For distinguishing LU-orbits, we have the following lemma:

\begin{Lemma} \label{lem:samerank}
    Consider two labeled graphs $G=(V,E)$ and $G'=(V,E')$ defined on the same node set $V$. If the corresponding graph states $\ket{G}$ and $\ket{G'}$ are LU-equivalent, we have
    \begin{equation}
        \fixeddimensionsset{k}{G}{i} = \fixeddimensionsset{k}{G'}{i}
    \end{equation}
    for all marginal sizes $k \in \{0, \dots, n \}$ and marginal dimensions $i \in \{0, \dots, k - 1\}$. Similarly, we have
    \begin{equation}
        \ranktensor{G}{k} = \ranktensor{G'}{k}.
    \end{equation}
\end{Lemma}

Note that \cref{lem:samerank} gives only a necessary condition for LU-equivalence, but not a sufficient one.
Note further that computing $\fixeddimensionsset{k}{G}{i}$ for all $i \in \{0, \dots, k - 1\}$ and computing $ \ranktensor{G}{k}$ provides the same information.
When the graphs have a marginal state with different dimensions, we can conclude they do not belong to the same LU-orbit. The converse is not true: two graphs with exactly the same marginal ranks are not necessarily LU-equivalent. The smallest counter-example are two graphs of $9$ nodes, we refer to \cref{sec:examples} for more details.

As an example of the usage of \cref{lem:samerank}, the left and middle graphs in \cref{fig:twobodymarginalLUandnonLU} belong to the same LU-orbit, and thus have the same marginal ranks. The right graph is not LU-equivalent: the two highlighted marginals have $\stabdim{\{1,3\}} = \stabdim{\{2,3\}} = 0$, whereas the other two graphs have $\stabdim{\{1,3\}} = \stabdim{\{2,3\}} = 1$.

\begin{figure}
    \centering
    \includegraphics[width = \linewidth]{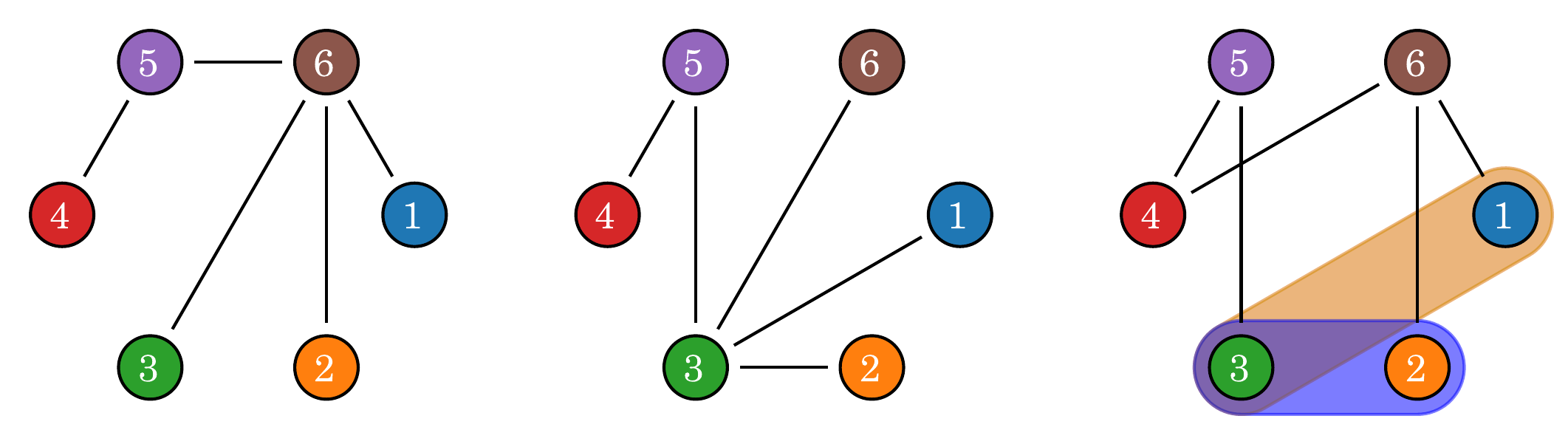}
    \caption{The graph on the left and in the middle have the same two-body marginal dimensions $\stabdim{\abs{M} = 2}$;  they also belong to the same LU-orbit (more specifically, they are related by a local complementation on node  6 (brown), followed by a local complementation on node 3 (green)). The graph on the right has certain marginals with a different dimension from the other two graphs: the highlighted two-body marginals $M_1 = \lbrace 1,3 \rbrace$ and $M_2 = \lbrace 2,3 \rbrace$ have $\stabdim{M} = 0$ for the right graph, instead of $\stabdim{M} = 1$ for the other two graphs.}
    \label{fig:twobodymarginalLUandnonLU}
\end{figure}

It can happen that two different orbits give the same two-body marginal ranks, but that higher-order marginals are different. In telling apart different LU-orbits, it is therefore sometimes necessary to increase the marginal size. For instance, the dimensions of all the $2$-body marginals from both graphs in \cref{fig:example_3marg} are  the same. However, some of the $3$-body marginals have different dimensions; the left graph has $\stabdim{M} = 1$ on the highlighted $3$-body set, whereas the right graph (corresponding to the absolutely maximally entangled state
of six qubits) has $\stabdim{M} = 0$ for the same set.

\cref{lem:samerank} is considerably easier to check for labeled graphs; for unlabeled graphs we additionally need to verify that we are comparing sets of associated nodes; there are generally a super-exponential number of these. 
However, there are ways to relax \cref{lem:samerank} so that it becomes more suitable for unlabeled graphs. We already have that \ranklist{}{k} is invariant under permutations of the nodes of the underlying graph, and as such constant when evaluated over elements in an entanglement class. However, \ranktensor{}{k} will generally not be invariant under these permutations. We first derive a permutation-invariant measure from the tensor \ranktensor{}{k}:
\begin{Definition}
    Consider a marginal tensor \ranktensor{G}{k} as defined in \cref{def:ranktensor}.  Let $\{\lambda_{1}, \dots, \lambda_{k}\}$ be the (real) eigenvalues of the Hermitian matrix we get after summing over $k-2$ arbitrary axes of the tensor. Then we define \tensoreig{G}{k} as the product of nonzero eigenvalues $\lambda_{i}$:
    \begin{equation}
        \tensoreig{G}{k} = \Pi_{\lambda_{i} \not = 0} \lambda_{i}.
    \end{equation}
    Note that the eigenvalues are invariant under permutation of the marginal tensor \ranktensor{G}{k}.
\end{Definition}

We are now ready to state a corollary of \cref{lem:samerank} containing two LU-invariant measures that are equal up to permutations of the nodes, i.e.~constant over elements of an entanglement class.

\begin{Corollary} \label{cor:rankslist}
    Consider two unlabeled graphs $G=(V,E)$ and $G'=(V,E')$ defined on the same node set $V$. If the corresponding graph states $\ket{G}$ and $\ket{G'}$ are LU-equivalent, we have
    \begin{align}
        \ranklist{G}{k} = \ranklist{G'}{k}, \label{eq:cond_ranklist} \\
        \tensoreig{G}{k} = \tensoreig{G'}{k}. \label{eq:cond_tensoreig} 
    \end{align}
\end{Corollary}

\begin{figure}
    \centering
    \includegraphics[width = 0.75\linewidth]{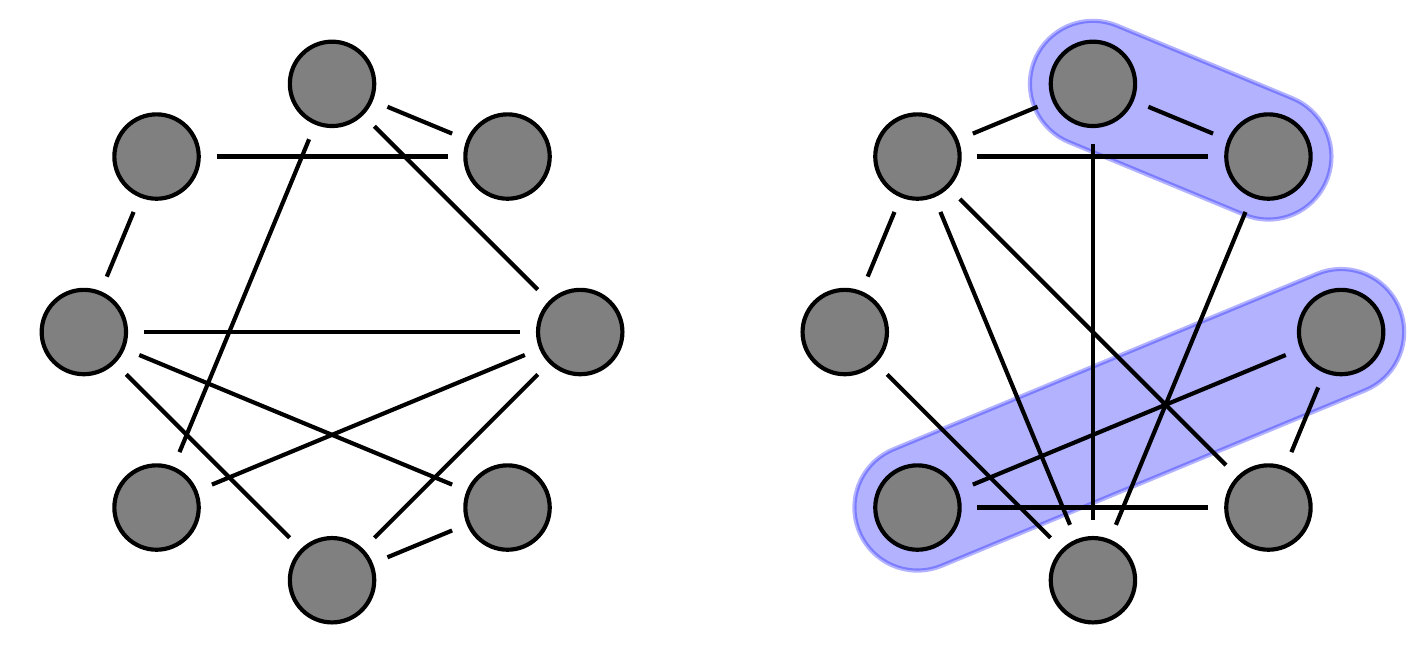}
    \caption{The unlabeled graphs, $G$ (left) and $G'$ (right) belong to different entanglement classes. The left graph has $\ranklist{G}{2} = [21, 0]$, whereas the right graph has $\ranklist{G'}{2} = [19,2]$; the two marginals with $\stabdim{} = 1$ have been highlighted. Note that purely the inequality $\ranklist{G}{2} \not = \ranklist{G'}{2}$ is enough to conclude that $G$ and $G'$ do not belong to the same entanglement class.}
    \label{fig:example_unlabeledringandline}
\end{figure}
Note that condition \eqref{eq:cond_ranklist} is easier to check in principle because it is sufficient to compute all marginal dimensions. However, checking condition \eqref{eq:cond_tensoreig} provides stronger results.
\cref{cor:rankslist} can be used to test if two unlabeled graph states belong to the same entanglement class. \cref{fig:example_unlabeledringandline} contains two unlabeled graphs that do not belong to the same entanglement class: they have a differently valued \ranklist{}{2}, and thus cannot belong to the same class.

We now investigate to which extent the knowledge of the dimensions of the complete set of $k$-body marginals of a labeled or unlabeled LU-orbit allows to characterize that orbit (or class). For a given labeled graph state, does its marginal tensor \ranktensor{}{k} uniquely correspond to its own LU-orbit, so that any graph state from another LU-orbit has a different \ranktensor{}{k}? Similarly, for an unlabeled graph state, does its (permutation-invariant) measure \ranklist{}{G} or \tensoreig{}{G} correspond uniquely to its own entanglement class?

To test how well the presented invariants perform in distinguishing different LU-orbits (for \ranktensor{}{k}) or classes (for \ranklist{}{k} and \tensoreig{}{k}), and comparing graphs for equivalence, we introduce two figures of merit; these can be computed for the aforementioned invariants, for any graph size $n$ and marginal size $k$:
\begin{itemize}
    \item The ratio of different invariants, divided by the total number of LU-orbits (for \ranktensor{}{k}) or entanglement classes (for \ranklist{}{k} and \tensoreig{}{k}). We denote this $r(\ranktensor{}{k})$, $r(\ranklist{}{k})$ or $r(\tensoreig{}{k})$.
    \item The probability that, given two random graph states $\ket{G}$ and $\ket{G'}$, their invariants are evaluated at the same value, even though they do not belong to the same LU-orbit (for \ranktensor{}{k}) or entanglement class (for \ranklist{}{k} and \tensoreig{}{k}). We denote these $p(\ranktensor{}{k})$, $p(\ranklist{}{k})$ and $p(\tensoreig{}{k})$.
\end{itemize}

We compute the marginal tensors of a representative of every LU-orbit for $k \in  \{2,\dots \lfloor \frac{n}{2}\rfloor\}$, up to $9$ qubits. Entanglement class representatives are sourced from the supplementary material of \cite{cabelloOptimalPreparationGraph2011a}; the representatives of the associated orbits are computed from these by exhaustive search through all permutations. 

Subsequently, we compute the figures of merit $r(\ranktensor{}{k})$ and $p(\ranktensor{}{k})$. If $r(\ranktensor{}{k}) = 1$ it means that every LU-orbit has a unique tensor \ranktensor{}{k} to identify it, whereas it is one divided by the number of orbits when all different LU-orbits have the same marginal tensor \ranktensor{}{k}. Similarly, if $p(\ranktensor{}{k}) = 0$ it means that two random LU-inequivalent graph states will always have a different \ranktensor{}{k}, while when $p(\ranktensor{}{k})$ is close to one, it is very likely that the two graph states will have the same \ranktensor{}{k}. See \cref{tab:discerningLUorbits} for the results. 

\begin{table}[htbp]
\centering
\begin{tabular}{l|ccc|ccc}
	 \toprule
 $n$  	& $r(\ranktensor{}{2})$ 	& $r(\ranktensor{}{3})$ 	& $r(\ranktensor{}{4})$ 	& $p(\ranktensor{}{2})$ 	& $p(\ranktensor{}{3})$ 	& $p(\ranktensor{}{4})$ \\ 
	 \midrule
3   	& 1	 & -	 & -	 & 0	 & -		 & -	   \\ 
4   	& 1	 & -	 & -	 & 0	 & -		 & -	   \\ 
5   	& 1	 & -	 & -	 & 0	 & -		 & -	   \\ 
6   	& 0.52	 & 1	 & -	 & 0.05	 & 0	 & -	   \\ 
7   	& 0.13	 & 1	 & -	 & 0.12	 & 0	 & -	   \\ 
8   	& 0.02	 & 0.88	 & 1	 & 0.22	 & 0.0001	 & 0   \\ 
9   	& 0.001	 & 0.48	 & 0.999	 & 0.37	 & 0.0004	 & 3e-10   \\ 
\bottomrule
\end{tabular}%
\caption{For labeled graphs of size $n \geq 3$ and marginal size $k \leq \lfloor \frac{n}{2}\rfloor$, we compute $r(\ranktensor{}{k})$ and $p(\ranktensor{}{k})$. 
If $r(\ranktensor{}{k}) = 1$, each LU-orbit uniquely corresponds to a specific \ranktensor{}{k}, serving as its identifier. Similarly, if $p(\ranktensor{}{k}) = 0$, any two LU-inequivalent graph states will have a different \ranktensor{}{k}.
}%
\label{tab:discerningLUorbits}
\end{table}

Similarly to the LU-orbits, we compute \ranklist{}{k} and \tensoreig{}{k} for a representative of every entanglement class, and compute the figures of merit $r(\ranklist{}{k})$, $r(\tensoreig{}{k})$, $p(\ranklist{}{k})$, and $p(\tensoreig{}{k})$. These results are given in \cref{tab:discerningclasseswithl} for \ranklist{}{k} and in \cref{tab:discerningclasseswitht} for \tensoreig{}{k}.

\begin{table}[htbp]
  \centering
\begin{tabular}{l|cccc|cccc}
	 \toprule
 $n$  	& $r(\ranklist{}{2})$ 	& $r(\ranklist{}{3})$ 	& $r(\ranklist{}{4})$ 	& $\mathbf{R}$  	& $p(\ranklist{}{2})$ 	& $p(\ranklist{}{3})$ 	& $p(\ranklist{}{4})$ 	& $\mathbf{P}$  \\ 
	 \midrule
3   	& 1	 & -	 & -	 & 1	 & 0	 & -		 & -		 & 0   \\ 
4   	& 1	 & -	 & -	 & 1	 & 0	 & -		 & -		 & 0   \\ 
5   	& 1	 & -	 & -	 & 1	 & 0	 & -		 & -		 & 0   \\ 
6   	& 0.73	 & 0.82	 & -	 & 1	 & 0.01	 & 0.01	 & -		 & 0   \\ 
7   	& 0.42	 & 0.85	 & -	 & 0.92	 & 0.17	 & 0.03	 & -		 & 0.03   \\ 
8   	& 0.15	 & 0.54	 & 0.56	 & 0.94	 & 0.30	 & 0.05	 & 0.03	 & 0.01   \\ 
9   	& 0.04	 & 0.34	 & 0.70	 & 0.83	 & 0.44	 & 0.05	 & 0.01	 & 0.01   \\ 
    \bottomrule
    \end{tabular}%
    \caption{
    The table details the computation of $r(\ranklist{}{k})$ and $p(\ranklist{}{k})$ for unlabeled graphs with sizes $n \geq 3$ and marginal sizes $k \leq \lfloor \frac{n}{2}\rfloor$, applying prior definitions for $r$ and $p$. It identifies entanglement classes using \ranklist{}{k} as unique markers or distinctions between classes. The $\mathbf{R}$ and $\mathbf{P}$ columns show the aggregated  ratios $r$ and  probabilities $p$ for class identification.
    }%
  \label{tab:discerningclasseswithl}%
\end{table}%

\begin{table}[htbp]
  \centering
\begin{tabular}{l|cccc|cccc}
	 \toprule
 $n$  	& $r(\tensoreig{}{2})$ 	& $r(\tensoreig{}{3})$ 	& $r(\tensoreig{}{4})$ 	& $\mathbf{R}$  	& $p(\tensoreig{}{2})$ 	& $p(\tensoreig{}{3})$ 	& $p(\tensoreig{}{4})$ 	& $\mathbf{P}$  \\ 
	 \midrule
3   	& 1	 & -	 & -	 & 1	 & 0	 & -		 & -		 & 0   \\ 
4   	& 1	 & -	 & -	 & 1	 & 0	 & -		 & -		 & 0   \\ 
5   	& 1	 & -	 & -	 & 1	 & 0	 & -		 & -		 & 0   \\ 
6   	& 0.73	 & 1	 & -	 & 1	 & 0.01	 & 0	 & -		 & 0   \\ 
7   	& 0.46	 & 1	 & -	 & 1	 & 0.16	 & 0	 & -		 & 0   \\ 
8   	& 0.19	 & 0.89	 & 1	 & 1	 & 0.30	 & 0.0001	 & 0	 & 0   \\ 
9   	& 0.06	 & 0.73	 & 0.998	 & 0.998	 & 0.44	 & 0.01	 & 1e-06	 & 1e-06   \\ 
    \bottomrule
    \end{tabular}%
    \caption{
    The table calculates $r(\tensoreig{}{k})$ and $p(\tensoreig{}{k})$ for unlabeled graphs of size $n \geq 3$ with marginal size $k \leq \lfloor \frac{n}{2}\rfloor$, leveraging earlier definitions for $r$ and $p$. It distinguishes entanglement classes using eigenvalues \tensoreig{}{k} as unique identifiers or indicators of different classes. Columns $\mathbf{R}$ and $\mathbf{P}$ present combined ratios $r$ and probabilities $p$ for class identification.
    }%
  \label{tab:discerningclasseswitht}%
\end{table}%

We see that for labeled graphs, our methods can perfectly distinguish all LU-orbits up to eight qubits, provided that marginals of a large enough size are used. At the same time, the number of different invariants versus the number of different orbits (i.~e.~$r(T_{k})$) diverge fast: e.g.~$r(T_{2})$ equals $1$ for up to five qubits, but becomes $0.001$ for eight qubits. 

For nine qubits and more, the methods provide necessary but not sufficient criteria.
 This is a direct consequence of the fact that for nine qubits there are (up to permutations) two separate orbits that have the exact same marginal structure. 
 
 Comparing the marginal structure allows us to identify counterexample candidates for the LU-LC conjecture. While there are 440 classes of unlabeled 9 node graphs, our method filters out 2  different classes, which are not LC-equivalent but might be LU-equivalent.
These classes are discussed in more detail in \cref{sec:examples}, where it is proven that for up to 10 qubits they are in fact LU-inequivalent in addition to being LC-inequivalent.

Turning to the second figure of merit (i.e. the right side of  \cref{tab:discerningLUorbits}), we also see that for up to 8 qubits there is always a marginal size such that the probability of a false positive becomes zero. This is a direct consequence of the left side of the table, more specifically that $r(T_{k})$ can equal $1$ for all these cases. Here, a false positive is the case that two (randomly chosen) graphs are identified as belonging to the same orbit, without that being the case.
The results differ, however, when the ratios and probabilities are not zero. For instance, for 7 qubits $r(T_{2})$ is only $0.13$, but still there is only a twelve percent probability that two random but non-LU-equivalent graphs have the same $T_{2}$. This disparity is explained by the fact that the orbits differ in size: the smallest orbit is the GHZ orbit with 8 graphs, of which two elements are shown in \cref{fig:example_starvscomplete}. The largest orbit, of which there are 105 permutational copies, contains 1096 graphs. This effect seems to get stronger with larger graphs. For 9 qubits, the number of different $T_{2}$'s is only $0.1\%$ of the total number of orbits. Still, the probability of a false positive is less than half. Similarly, $r(T_{3})$ is roughly half, but the chance of a false positive is negligible. 

For unlabeled graphs (i.e.~entanglement classes) we see a similar behaviour. First considering $l_{k}$, we see similar scaling as for labeled graphs: the number of different invariants versus the number of different classes diverge fast, and the probabilities for a false positive increase slower. Only classes up to 6 qubits can be perfectly distinguished, which is less than for the orbits. Nevertheless, the same behaviour as for labeled graphs is obtained when considering $t_{k}$ instead: here, again all classes up to and including 8 qubits can be perfectly distinguished. 

There is one surprising property that exists when comparing lists ($\ranklist{}{k}$). We note that the reliability of our methods increases, if we consider lists of several marginal set sizes $k$ together at the same time. For example for 8-node graphs, the ratio $r_{k}$ is below $60\%$ for each $k$ equal to 2, 3, and 4. However, combining all three lists gives a ratio of $94\%$. One reason is that the information given in a list $\ranklist{}{k}$ can only be partly estimated by a list of higher $k' > k$. This behaviour is not apparent in $T_{k}$ (and by extension $t_{k}$), because $T_{k-1}$ is `embedded' into the (higher-dimensional) diagonal of $T_{k}$.

\subsection{Characterizing LU-orbits by Their Graph Structures} \label{sec:graphstructures}

A visual tool to distinguish graph states are metagraphs. We can directly derive the set $\SSS_M$  of a marginal set $M$ from its metagraph. 
We can further group metagraphs into metaorbits to see which graph structures transform into each other.
In this section we discuss how to derive $\SSS_M$ for general sets $M$. We further present the metagraph orbits of sets of size two and three and discuss how our observations are connected to the results of the previous section.

\subsubsection{Computing the Marginal Dimension Graphically}

A metagraph of a set $M = \lbrace 1, 2 \rbrace$ has three type-2 nodes: $\mathcal{P}(\lbrace 1, 2 \rbrace) \setminus \emptyset = \lbrace \{1\}, \{2\}, \{1,2\} \rbrace$. The marginal $\rho_M$ can at most get stabilized by $\1, g_1, g_2$, and $g_1 g_2$, since all other stabilizer elements $g_j$ for $j \notin M$ vanish in the partial trace. We determine which of the stabilizer elements do not vanish from the non-empty neighbor sets. The stabilizer element $g_1$ does not vanish if the type-2 nodes $[1]$ and $[1,2]$ are not connected to type-1 nodes in $M$. 
This is the only case where $g_1$ has support on $M$ and no support outside, that is $g_1 = X_1 Z_2$. An equivalent argumentation can be done for $g_2$. 
 The stabilizer element $g_1 g_2$ does not vanish if the type-2 nodes $[1]$ and $[2]$ are not connected to type-1 nodes in $M$. Depending on whether type-1 nodes 1 and 2 are connected or not, $g_1 g_2$ is given by 
$Y_1 Y_2$ 
 or 
 $X_1 X_2 $, respectively, which has support on $M$. 
The case of two node sets was already discussed in 
Ref.~\cite{Gittsovich2010Multiparticle}. An example is shown in \cref{fig:example_metagraphs}(d)  where the lower right (yellow) metagraph has $[5,7]$ as only type-2 node connected to type-1 nodes in $M = \{ 5,7 \}$. This indicates that the marginal $\rho_{\{ 5,7 \}}$ is stabilized by $\SSS_{\{ 5,7 \}} = \{ \1, g_5 g_7 \}$. It follows that $d_{\{ 5,7 \}} = 1$.

This concept can get generalized to an arbitrary set $M$.
The marginal state $\rho_M$ can in principle get stabilized by stabilizer elements $\lbrace g_L \mid L \subseteq M \rbrace$, where $g_L \coloneqq \prod_{i \in L} g_i$. A stabilizer element $g_L$ has support on $M$ if and only if all type-2 nodes which are connected to $M$ are either connected to an even number of type-1 nodes in $L$ or are not connected to nodes in $L$. That is $g_L$ is in $\SSS_M$ if and only if all type-2 nodes $w \in \mathcal{P}(M)_{k \geqslant 0}$ for which $\vert w \cap L \vert$ is an odd number are not connected to type-1 nodes in $M$.
Note that a similar definition was given in Ref.~\cite{claudet2024covering}.

For example the metagraph (b) in \cref{fig:example_interesting_marginals} has two type-2 nodes $\{ [1 ,2], [2,3] \}$ connected to the set $M = \{ 1,2,3 \}$. This graph gets stabilized by $g_{L = \{ 1,2,3 \}} = g_1 g_2 g_3$, since $\abs{\{1 ,2 \} \cap \{1 ,2, 3 \}} = \abs{\{2,3 \} \cap \{1 ,2, 3 \}} = 2$. All other subsets $L \subseteq M$, which are not the empty set have an odd number of elements in the intersection with $\{1 ,2 \} $ or $ \{2,3 \}$ and therefore there are no other stabilizers than $\1$. 

\begin{figure}
    \centering
    \includegraphics[width = 0.95\linewidth]{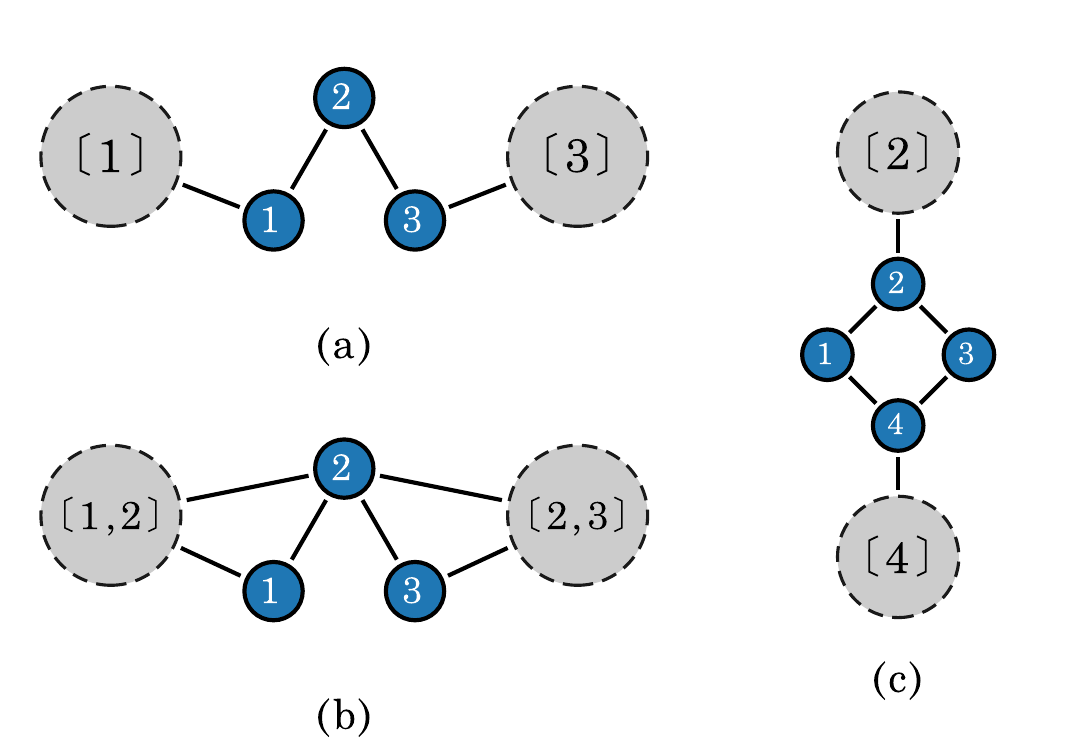}
    \caption{Example of new structures detected by marginals. \textbf{$(a)$} We have $\SSS_{\lbrace 1,2,3 \rbrace} = \lbrace \1, g_2 \rbrace$, while $\SSS_{N \subsetneq \lbrace 1,2,3 \rbrace} = \lbrace \1 \rbrace$ for all proper subsets $N$. Therefore we can detect this structure with three-marginals, but not with two-marginals. \textbf{$(b)$} We have $\SSS_{\lbrace 1,2,3 \rbrace} = \lbrace \1, g_1 g_2 g_3 \rbrace$, while $\SSS_{N \subsetneq \lbrace 1,2,3 \rbrace} = \lbrace \1 \rbrace$ for all proper subsets $N$. Also here we can detect this structure with three-marginals, but not with two-marginals. \textbf{$(c)$} We have $\SSS_{\lbrace 1,2,3, 4 \rbrace} = \lbrace \1, g_1, g_4, g_1 g_4 \rbrace$, $\SSS_{\lbrace 1, 4 \rbrace} = \lbrace \1,  g_1 g_4 \rbrace$ and $\SSS_{N \subsetneq \lbrace 1,2,3,4 \rbrace} = \lbrace \1 \rbrace$ for all other proper subsets $N$. Therefore we can detect some details of the structure with four-marginals, but not with lower number marginals.}
    \label{fig:example_interesting_marginals}
\end{figure}

An alternative definition of the stabilizer set (already defined in \cref{eq:stab_set_tr,eq:stab_set_supp}) is given by
\begin{align}
    \SSS_M &\coloneqq \left\{ g_{L \subseteq M} \middle|
    \begin{array}{l}
     \forall w \in \mathcal{P}(M)_{k \geqslant 0}\text{ connected to } M : \\ 
     \vert L \cap w \vert \text{ is even}  
    \end{array} 
    \right\}. \label{eq:stab_set_meta}
\end{align}
Therefore we can reproduce the results discussed in the previous chapter from metagraphs.

\subsubsection{Marginal Orbits}

A deeper look into graph structures can be taken by considering marginal orbits. 

\begin{Definition}[Marginal orbit]
    Two metagraphs $G_M$ and $G'_M$ belong to the same marginal orbit, if there are LU-equivalent graphs $G$ and $G'$ and a set $M$ such that $G_M$ and $G'_M$ are their marginal graphs.
\end{Definition}

Note that a marginal orbit is less restrictive than a graph orbit. Two LU-inequivalent graphs might have the same metagraph. 

\begin{Corollary} \label{cor:cluster}
     Consider two graphs $G=(V,E)$ and $G'=(V,E')$ defined on the same node set $V$. If the corresponding graph states $\ket{G}$ and $\ket{G'}$ are LU-equivalent, all pairs of metagraphs belong to the same marginal orbit.
\end{Corollary}

Marginal orbits of two node sets as well as clustering of two node sets with certain properties were discussed in Ref.~\cite{burchardt2023foliage} (in this reference, it was called the foliage partition).
It was shown that metagraphs of sets $M$ of two nodes and $\stabdim{M} = 1$ belong to the same marginal orbit. It was further shown that metagraphs of general sets $M$ and $\stabdim{N} = 1$ for all $N \subseteq M$, $\vert N \vert = 2$ belong to the same marginal orbit.

While the metaorbit of two-node metagraph states $G_M$ can be decided by computing $\stabdim{M}$, this is not the case in general. For example for $\vert M \vert = 3$, $\stabdim{M}$ can be 0, 1, or 2, but there are more than three orbits. We can have   $\stabdim{M} = 1$ in two cases: either because we found a two-marginal and added one node, that is  $\stabdim{N} = 1$ for one $N \subsetneq M$ or because we found a new structure, that is $\stabdim{N} = 0$ for all $N \subsetneq M$. The two cases clearly belong to different marginal orbits.
Also, we can have   $\stabdim{M} = 2$ in two cases: either because we found a two-marginal and a new structure in three nodes, that is  $\stabdim{N} = 1$ for one $N \subsetneq M$ or because all subsets of two nodes have an interesting structure, that is $\stabdim{N} = 1$ for all $N \subsetneq M$. Also here the two cases clearly belong to different marginal orbits.

We can determine whether $\stabdim{M} > 0$ indicates a new structure on the whole set $M$ or whether all structures were already detected in subsets of $M$ by comparing stabilizers.  
We find something new, if and only if the following relation holds:
\begin{align}
    \vert \SSS_M \vert &> \vert \langle \bigcup_{N \subsetneq M} \SSS_N \rangle \vert,  \label{eq:threecondset}
\end{align}
where $\langle \SSS \rangle$ denotes the group generated by elements in $\SSS$.
In the case $\vert M \vert = 3$, an equivalent condition is given by
\begin{align}
    \stabdim{M} &>  \sum_{N \subset M, \vert N \vert = 2} \stabdim{N}. \label{eq:threecond} 
\end{align}
In \cref{fig:example_interesting_marginals} there appear examples of structures which fulfill this condition.

An observation about \cref{eq:threecondset} is stated in the following lemma. It is closely related to Lemma $1$ from \cite{nestLocalUnitaryLocal2005} and mostly follows from it. For completeness we still give a self-contained proof.

\begin{Lemma}
\label{Lemma:LHS_vs_RHS}
    For any set $M$, the relation between the dimensions of a stabilizer set of a marginal state $\rho_M$ and its sub-marginals obey the following relation:
    \begin{align}
    \label{eq:LHS_vs_RHSmain}
        \vert \SSS_M \vert = 2^\ell \vert \langle \bigcup_{N \subsetneq M} \SSS_N \rangle \vert,
    \end{align}
    where $\ell \in \lbrace 0, 1, 2 \rbrace$. 
    Furthermore, if \cref{eq:LHS_vs_RHSmain} is satisfied for $\ell=2$, then the size of $M$ is an even number, $\langle \bigcup_{N \subsetneq M} \SSS_N \rangle = \{I\}$ and  it holds that $\abs{\SSS_{M}} = 4$.
\end{Lemma}

\cref{app:Lemma2_proof} contains a proof of the statement. Notice that this relation was recently observed in Ref.~\cite{claudet2024covering} for the special case when $\vert \langle \bigcup_{N \subsetneq M} \SSS_N \rangle \vert=1$.

A straightforward partition into orbits can be done by computing the stabilizer dimension of the set $M$ and all subsets. A finer partition might be found numerically. Interestingly, for three-node metagraphs the finest numerical partition coincides with the partition found by stabilizer dimensions. We might see a finer partition for larger sets.

\section{Distinguishing LU-Orbits Beyond the Tool of Marginal Dimensions} \label{sec:examples}

Our methods can distinguish LU-orbits and classes of graphs with up to 8 nodes. 
The first example where all signatures are identical but the graphs are not LU-equivalent are two classes of graphs with 9 nodes. Representatives of these are shown in \cref{fig:example_9qubit}. The Bouchet algorithm shows that these two graphs are not LC-equivalent~\cite{bouchet1993Recognizing}. In the following lemma we show that they are not LU-equivalent either.

\begin{Lemma}
    The graphs $L$ and $R$ are not LU-equivalent.
\end{Lemma}

\begin{proof}
    
Let us compare the stabilizer of marginal $\rho_{\{1,2,3,5\}}$ for both $L$ and $R$. We find them to be stabilized by
\begin{align}
\SSS_{\{1,2,3,5\}}^{L} &= \{\Id,g_2^L\} =\{ \Id, Z_1 X_2 Z_3 Z_5 \}\label{stab_1},\\
\SSS_{\{1,2,3,5\}}^{R} &= \{\Id,g_{1}^R g_{2}^R g_{3}^R g_{5}^R\} = \{\Id,- Y_{1} Y_{2} Y_{3} Y_{5}\}, \label{stab_2}
\end{align}
for $L$ and $R$, respectively; the last equality holds since $XZ = -iY$.

Assume now that there exists a local unitary operation $U = U_1 U_2 U_3 \cdots U_n$ such that $U \ket{L} = \ket{R}$. 
Notice that 
\[
U_1 U_2 U_3 U_5 \,\rho_{\{1,2,3,5\}}^L \,U_1^\dagger U_2^\dagger U_3^\dagger U_5^\dagger  = \rho_{\{1,2,3,5\}}^R ,
\]
and by \cref{eq:AdNew} and comparing \cref{stab_1} and \cref{stab_2}, we then know that $U_1 Z U_1^\dagger = Y$ up to a phase.

On the other hand, for the stabilizer of the marginal $\rho_{\{1,4,5,7\}}$ we find:
\begin{align*}
\SSS_{\{1,4,5,7\}}^{L} &= \{ \Id,g_{4}^{L} \}= \{\Id,Z_{1} X_{4} Z_{5} Z_{7} \}\\
\SSS_{\{1,4,5,7\}}^{R} &= \{\Id,g_{4}^{R} \}= \{\Id,Z_{1} X_{4} Z_{5} Z_{7}\}.
\end{align*}
Hence, $U_1Z U_1^\dagger=Z$ up to a phase. 
This is a contradiction to $U_1Z U_1^\dagger = Y$ up to a phase.
\end{proof}

\begin{figure}
    \centering
    \includegraphics[width = 0.8\linewidth]{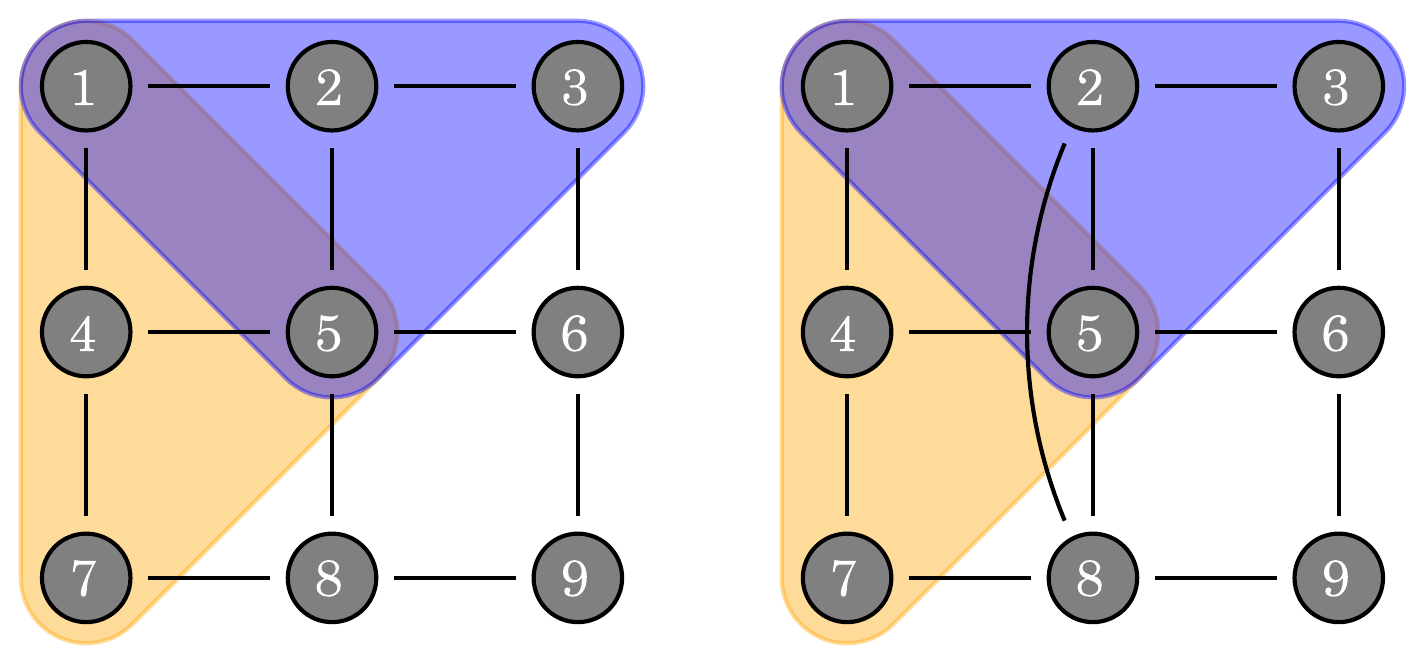}
    \caption{Two nine-qubit graph states $L$ (left) and $R$ (right) that are not LU-equivalent. 
    The marginals $\rho_{\{1,2,3,5\}}$ and $\rho_{\{1,4,5,7\}}$ are highlighted in blue and orange, respectively.
    }
    \label{fig:example_9qubit}
\end{figure}

We have found other examples of graphs that have the same signatures but are not LC-equivalent, e.g.~those shown in \cref{fig:example_peterson}. 
Notably, the example in that figure shows that elements from different LU-orbits but the same LU-class can have the exact same marginal structure.
These are potential candidates for counter-examples to the LC-LU conjecture with a relatively low number of qubits. However, checking all equivalence classes of isomorphic LC-orbits, there are no other examples of 9 nodes, and all examples of 10 nodes can be shown to be LU-inequivalent using the method described in this section. Thereby, we increased the bound for the maximum number of nodes for which the LU-equivalence of unlabeled graphs is identical to the LC-equivalence from 8 to 10.


\begin{figure}
    \centering
    \includegraphics[width = 0.8\linewidth]{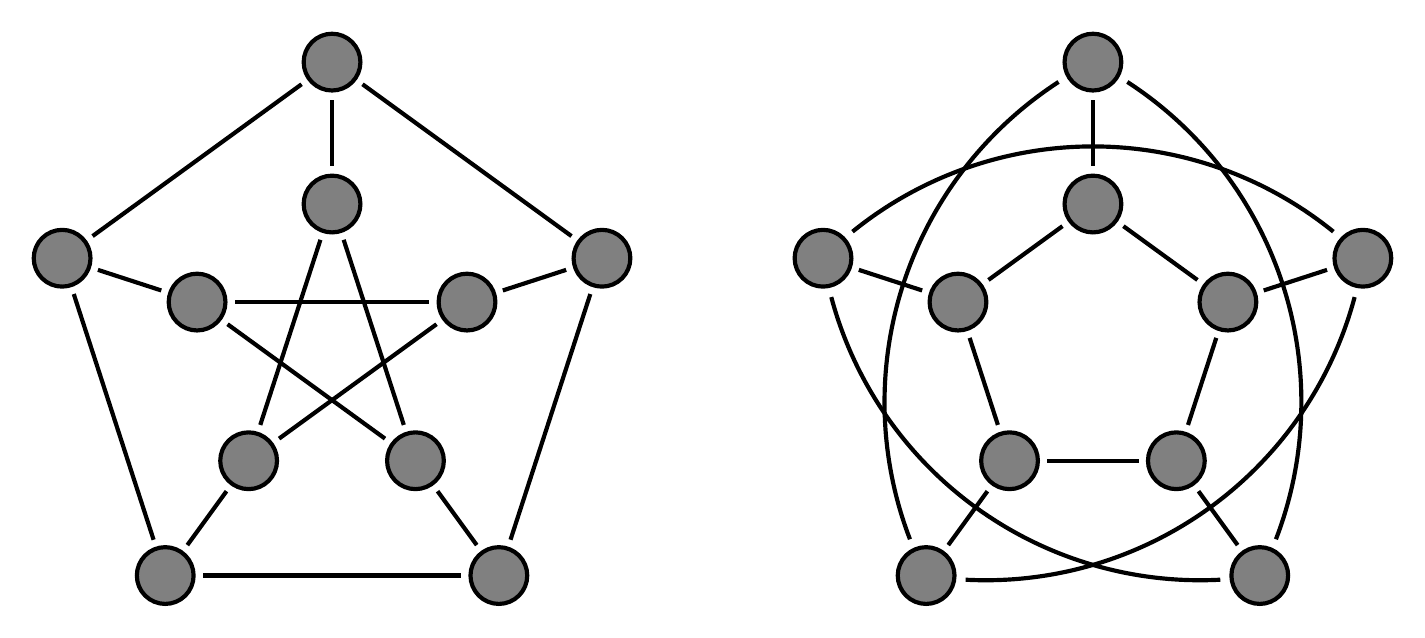}
    \caption{The graph on the left, known as the \textit{Peterson graph}, and the graph on the right have the same \ranktensor{}{k} for all $k$. The graphs are related by permutation of the `inner' and `outer' nodes, so they belong to the same class. However, they do not belong to the same $LC$-orbit.
    }
    \label{fig:example_peterson}
\end{figure}

\section{Tools for characterizing LC-orbits: Condensed Graphs} \label{sec:lctools}

With the tools introduced in the previous section, we can detect many structures in graphs. However, due to limits in computational power, we are restricted to certain marginal orders.  In this section we discuss how to condense graphs such that we can detect some of the structures with marginals of lower order. Condensed graphs were introduced in Ref.~\cite{burchardt2023foliage} as foliage graphs together with a specific condensation rule. Ref.\ \cite{zhang2023bell} showed that the condensation rule can get relaxed for labeled graphs. After discussing the previous results, we present a generalization of condensing graphs for labeled and unlabeled graphs.
It should be stressed that the tools derived in this section can only prove inequivalence under LC operations, 
and not under more general LU-operations. 

By a condensed graph, we mean a graph where we combine a set of nodes of the initial graph into one node. The adjacency between the new node and the rest of the graph is derived by the adjacency of the set of nodes that got condensed into one node and the rest of the graph. We define a condensed graph as follows:

\begin{Definition}
    Consider a graph $G=(V,E)$ and a set $C \subseteq V$. The \textit{condensed graph} $G_C =(V_C, E_C)$  consists of the node set $V_C = \{ c \} \cup ( V \setminus C)$
    and edge set $E_C$ defined in the following way: $(i,j) \in E_C$ if either $i,j\in V \setminus C$ and $(i,j) \in E$, or $j=c$ and there exists $s\in C$ such that $(i,s) \in E$.
\end{Definition}
The definition of a condensed graph with respect to multiple sets $\CC = \lbrace C \rbrace$ can be formulated analogously. 

Some choices of condensation sets $C$ preserve LC-equivalence of graphs: 
\begin{Lemma} \label{lem:condensed}
    Consider two graphs $G$ and $G'$ and a two-node condensation set $C$ such that $\stabdim{C} = 1$. If 
    $G$ and $G'$ are LC-equivalent, it follows that $G_c$ and $G'_c$ are LC-equivalent.
\end{Lemma}
Recall that $d_C$ is the marginal dimension of the marginal set $C$ as defined in \cref{eq:d_M}.
This statement was proven for the special case $\CC = \lbrace C \vert \stabdim{C} = 1 \rbrace$ in Ref.~\cite{burchardt2023foliage} and for general two-node sets in Ref.~\cite{zhang2023bell}. Note that we can 
use the negated statement to exclude LC-equivalence: If  the condensed graphs $G_c$ and $G'_c$ are LC-inequivalent, also the initial graphs  $G$ and $G'$ are LC-inequivalent.

While on labeled graphs, the condensation can be done on specific sets, for unlabeled graphs it is challenging to find the same sets on two nodes. 
Therefore one  has to  find extra conditions like condensing all sets with a certain property. Practical properties can be ``all sets $C$ of size two and $\stabdim{C} = 1$'', as in \cite{burchardt2023foliage} or   ``all sets $C$ of size two and $\stabdim{C} = 1$ which belong to a cluster of size $k$''.
An example is shown in \cref{fig:ex_cond}.

\begin{figure}
    \centering
    \includegraphics[width=.9\linewidth]{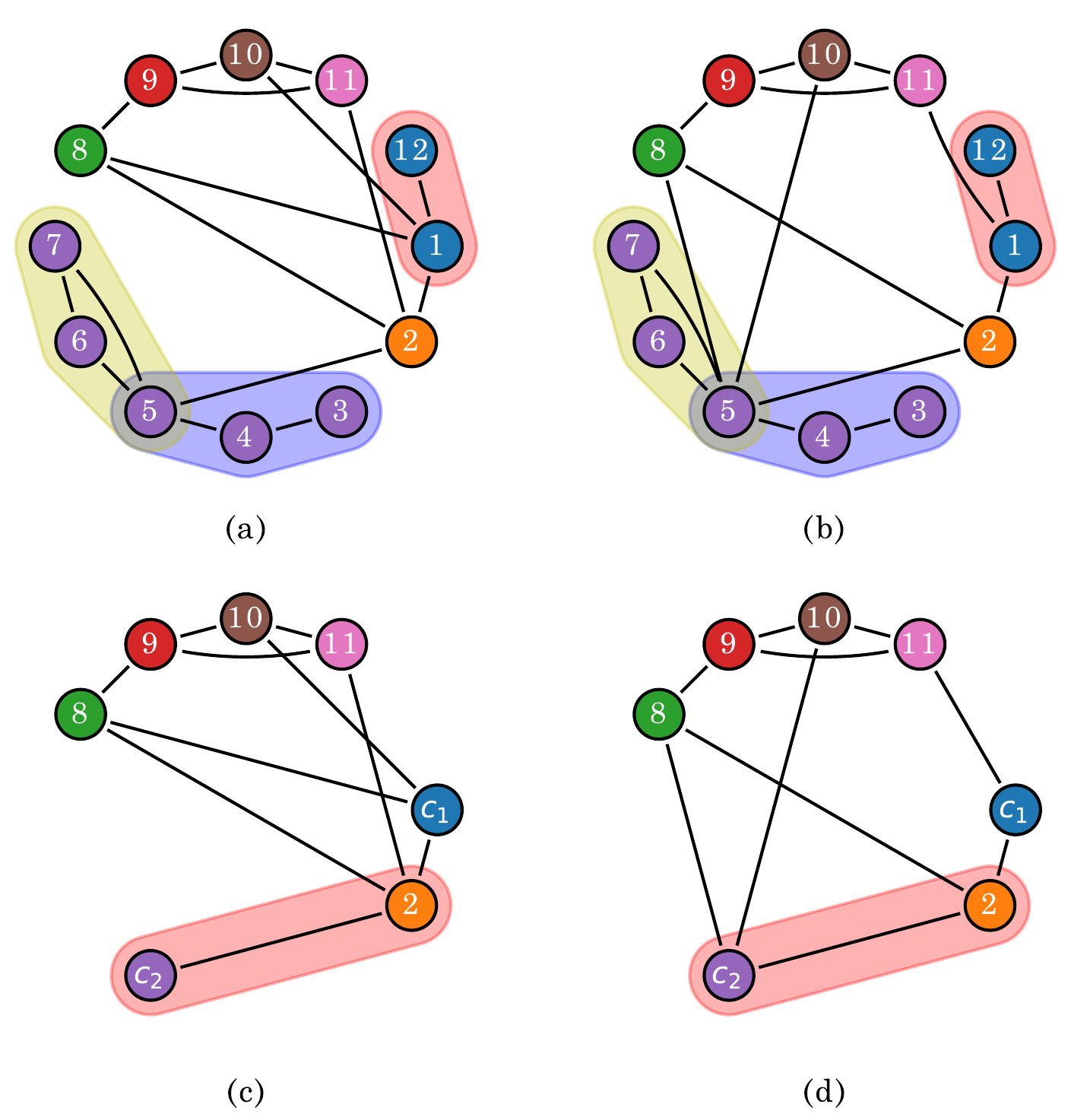}
    \caption{\textbf{$(a), (b)$} Two 12-node graphs which have all marginal dimensions of two and three node sets equal, that is $\fixeddimensionsset{2}{G}{i} = \fixeddimensionsset{2}{G'}{i}$ and $\fixeddimensionsset{3}{G}{i} = \fixeddimensionsset{3}{G'}{i}$ for all $i < 2$ and 3, respectively. \textbf{$(c), (d)$} After condensation of sets  $\{1,12\}$ and $\{3,4,5,6,7\}$, we see that the condensed 7-node graphs differ by their marginal dimension $\stabdim{2,c_2}$. Therefore, the two initial graphs are not LC-equivalent.}
    \label{fig:ex_cond}
\end{figure}

We show that the condensation sets can be chosen more generally. 
\begin{Lemma} \label{lem:condens_big}
    Given two graphs $G$ and $G'$ and a condensation set $C$ such that $\stabdim{C} = \vert C \vert - 1$.  If 
    $G$ and $G'$ are LC-equivalent, it follows that $G_c$ and $G'_c$ are LC-equivalent.
\end{Lemma}
Note that \cref{lem:condensed} is a special case of \cref{lem:condens_big}, where $\vert C \vert = 2$.
In the following we present a proof idea. For the complete proof, see \cref{sec:proof_condens}.

Proof idea:
We first show that every  set $C$ with $\stabdim{C} = \vert C \vert - 1$ has exactly one neighborhood. That is, there is a neighborhood $\NN{C}$ of nodes connected to nodes in $C$ such that $C$ can be composed into two disjoint sets of isolated nodes $C_I$ and nodes which are connected to the neighborhood $C_N$ such that for all $i \in C_N$ and all $n \in \NN{C}$ we have that $(i,n) \in E$.
We then show that for two graphs which are LC-equivalent also the condensed graphs are LC-equivalent.

There might be further conditions for condensation graphs, which we present in the following conjecture.

\begin{Conjecture} \label{lem:cond_neighb}
    Given two graphs $G$ and $G'$ and a condensation set $C$ such that each node in $C$ is connected to at most one node in the neighborhood in  $V\setminus C$. 
     If $G$ and $G'$ are LC-equivalent, it follows that $G_c$ and $G'_c$ are LC-equivalent.
\end{Conjecture}

It would be interesting to find more condensation rules. 
In \cref{sec:proof_condens} we show a list of some trials together with counter examples which show that they do not work.

\section{Computational Complexity} \label{sec:complex}

In this section we compare the presented methods to existing algorithms and discuss their computational complexity. Note that we compute the complexity of the algorithms we used, which is in general an upper bound to the complexity of the best algorithm. 

It is possible to decide LU-equivalence by solving a finite set of polynomial equations \cite{Kraus2010_PRA,Maciazek2013howmany}. However, there is no efficient algorithm known. 
The Bouchet algorithm which can be used to determine LC-equivalence of two graphs has a computational complexity of $\OO(n^4)$. Deciding pairwise LC-equivalence using Bouchet's algorithm scales quadratically in the number of graphs. The presented methods to decide LU-equivalence scale linearly in the number of graphs. The complexity of our methods for $k \geqslant 2$ is upper bounded by $\OO (k^{\frac{3}{2}-k} n^{k+1})$. However, for graphs of 9 or more nodes, our methods are only approximations.  

In \cref{sec:tools}, we presented the marginal dimension $\stabdim{M}$ of a subset of nodes $M$. Although calculating any marginal state of an arbitrary quantum state takes exponentially many steps, the complexity to compute $\stabdim{M}$ for a given set $\abs{M} = k$ on a graph of $n$ nodes is $ \OO ((n - k) k^2)$. There are $\binom{n}{k}$ subsets $M \subseteq V$ of size $k$, so by using the Stirling approximation we find that the complexity to compute $\stabdim{M}$ for all subsets $\abs{M} = k$ on a graph of $n$ nodes is $ \OO (k^{\frac{3}{2}-k} n^{k+1})$.

In \cref{sec:rankinvar}, we introduced the measures $\fixeddimensionsset{k}{G}{i}$, $\ranklist{G}{k}$ and $\ranktensor{G}{k}$. Their computational complexity is the same as computing $\stabdim{M}$ for all subsets $\abs{M} = k$. For the measure $\tensoreig{}{k}$, additionally the eigenvalues of an $n \times n$ Hermitian matrix needs to be computed, which in practise takes $\OO (n^{3})$. Therefore, to calculate $\tensoreig{}{k}$ given $\ranktensor{G}{k}$ takes $\OO (n^{3})$, while calculating $\tensoreig{}{k}$ without $\ranktensor{G}{k}$ takes $\OO (n^{3})$ for $k = 1$ and $ \OO (k^{\frac{3}{2}-k} n^{k+1})$ for $k \geqslant 2$.

In \cref{sec:lctools}, we discussed graph condensation. For the condensation rule given in \cref{lem:condens_big}, we computed the marginal dimensions $\stabdim{M}$ of all sets $M$ of a certain size $k \geqslant 2$. Then, we condense  the graph, which can be done with a loop over the adjacency matrix in $\OO(n^2)$ steps. Therefore, also for condensation, the complexity is  $ \OO (k^{\frac{3}{2}-k} n^{k+1})$. Note that condensed graphs have fewer nodes and are therefore computationally easier to compare. Condensed graphs can get compared using the other methods presented in this paper or get further condensed.

In practice, different methods should be combined. To decide LU-equivalence of two graphs, given the complexity of the methods, it is reasonable to first test for LC-equivalence using the Bouchet algorithm. 
If the graphs are LC-equivalent, LU-equivalence follows immediately. If the graphs are not LC-equivalent, LU-equivalence can be tested by computing marginal dimensions.
Only if the graphs are not LC-equivalent and LU-equivalence cannot be excluded by comparing the marginal dimension, more advanced techniques have to be applied. Options are graph condensation as introduced in \cref{sec:lctools},  comparing stabilizer operators as discussed in \cref{sec:examples}, or solving the set of polynomial equations from Ref.\ \cite{Kraus2010_PRL}.
For deciding LU-equivalence of a sufficiently large set of graphs, it is however recommendable to first compute the marginal dimensions of all graphs and test for LC-equivalence only if LU-equivalence cannot be excluded. This is due to the fact that our methods scale linearly in the number of graphs while Bouchet's algorithm scales quadratically.

\section{Summary and Outlook} \label{sec:sum}

This paper aims at advancing the understanding of the relationship 
between the marginal state properties and the LU-equivalence of 
graph states. For that, we introduced an invariant under local unitaries related to the entanglement entropy.  We showed how the invariant can be computed in several ways: from the adjacency matrix, the rank of marginal states, stabilizer properties, and also from geometric graph properties. We then demonstrated how it can be used to distinguish LU-orbits of graph states.

We have shown that the methods to distinguish LU-orbits work perfectly up to 8 qubits. There exist graphs of 9 or 10 qubits for which our methods fail to distinguish graphs, even though they are known to be in different LC-orbits. For unlabeled graphs,  we have shown their LU-inequivalence by another method, concluding that up to 10 qubits no graphs from different isomorphic LC-orbits are LU-equivalent.
It was previously known that there are no counterexamples to the LU-LC conjecture for unlabeled graph states of 8 qubits or smaller~\cite{cabelloEntanglementEightqubitGraph2009}. 
Our computations prove the LU-LC conjecture holds at least up to 10 qubits for unlabeled graphs. Furthermore, we found out that the bound of 8 is true for labeled graphs, which were not considered previously, as well. Most likely the LU-LC conjecture is even true up to 10 qubits in the labeled graphs, as the only possible counterexamples of that size would have to be isomorphic but LC-inequivalent, which is extremely unlikely.

Based on our results, it is both possible to distinguish graph state orbits by computer-implemented algorithms and (often) intuitively by visual inspection of their graphs. We further showed how the problem of deciding LC-equivalence can be mapped to condensed graphs with fewer nodes than the original graphs. 

The methods that we have presented make use of the concept of tensors \ranktensor{G}{k}. However, it is worth mentioning some alternative concepts such as the sector length \cite{Wyderka_2020} and cut-rank \cite{NGUYEN2020103183,10.1007/11940128_64} of a graph.

Further research could focus on gaining a better understanding of marginal orbits, particularly in terms of their advantages over computing marginal dimensions of sets with more than three nodes. Additionally, it would be interesting to discover more general condensation rules. 

In this paper, we restricted ourselves to local operations being unitary.
In general, local operations may include both unitary operations and measurements.  
Although the problem is addressed in several works \cite{Hahn_2019_network_routing,Mannalath2023multiparty,dejong2023extracting,brand2023quantum,szymanski2024useful},
the structural relationship between which orbits of states can be transformed into which other orbits of states is not yet well understood. 
It would be interesting to investigate whether our methods can be extended to this more general case.

\section*{Acknowledgments}
We thank  Jan L.\ B\"onsel, Kiara Hansenne, Lucas E.\ A.\ Porto, and Fabian Zickgraf for discussions. This work was supported by 
the Deutsche Forschungsgemeinschaft  (DFG project numbers 447948357, 440958198 and the Emmy Noether grant 418294583), the Sino-German Center for Research Promotion (Project M-0294), the ERC (Consolidator Grant 683107/TempoQ), the German Ministry of Education and Research (Project QuKuK, BMBF Grant No. 16KIS1618K), the Stiftung der Deutschen Wirtschaft, the European Union via the Quantum Internet Alliance project (Project ID 101102140), the BMWK-funded project Qompiler, and NWO Vidi grant (Project No VI.Vidi.192.109).
This work is co-funded by the European Union (ERC, ASC-Q,
101040624). Views and opinions expressed are however those of the authors
only and do not necessarily reflect those of the European Union or the
European Research Council. Neither the European Union nor the granting
authority can be held responsible for them.

\onecolumngrid
\appendix

\section{Orbits of Labeled and Unlabeled Graphs}\label{app:labelvsunlabel}

To illustrate the differences between labeled and unlabeled graphs, we present the four-qubit linear cluster state $\ket{L_{\mathrm{1234}}}$, shown in the upper left corner of \cref{fig:33graphslabeled}, as an example. 
For denoting line graphs, we use the notation $L_{\mathrm{1234}}$, where the index indicates the order of labels in the graph. That is, $L_{\mathrm{1234}}$ is a four node graph with edge set $\{(1,2),(2,3),(3,4)\}$.
In total, there are 33 different graphs that are associated with the line graph through local complementations and permutations of the nodes shown in \cref{fig:33graphslabeled}. Indeed, there are 12 distinct permutations of the line graph, and 21 other connected four-qubit graphs that can be obtained by a series of local complementations on at least one of these 12 distinct permutations. Depending on whether the graphs are considered as labeled or unlabeled graphs, they can be grouped differently into LU-orbits. The two different groupings are shown in \cref{fig:33graphslabeled,fig:33graphsunlabeled} and described in detail below.

When the 33 graphs are considered as labeled graphs, they fall into three distinct sets of LU-orbits, shown in \cref{fig:33graphslabeled}.  We note that certain permutations of the line graph, although distinct graphs, can be obtained from the graph $L_{\mathrm{1234}}$ by local complementations. Indeed, consider the graph $L_{\mathrm{1243}}$, shown on the 7th position of the first row in \cref{fig:33graphslabeled}. 
The graph $L_{\mathrm{1243}}$ is related to $L_{\mathrm{1234}}$ by a local complementation on node $4$ followed by a local complementation on node $3$, showing that the two associated graph states are in each others LU-orbit.

Not all permutations of $L_{\mathrm{1243}}$ fall into its associated LU-orbit. Indeed, the graph state $\ket{L_{\mathrm{1432}}}$ is not LU-equivalent to $\ket{L_{\mathrm{1243}}}$. This follows e.g.~from the fact that the LU-invariant stabilizer dimension $\stabdim{\{1,2\}}$, defined in \cref{eq:d_M}, is not the same for the two states. Hence, the graph state $\ket{L_{\mathrm{1432}}}$ gives rise to its own LU-orbit; which is equally large as the LU-orbit of $\ket{L_{\mathrm{1234}}}$ (a fact that follows from a symmetry argument).

This results in three separate LU orbits, that each contain some permutations of the linear cluster state $\ket{L_{\mathrm{1234}}}$. These three LU-orbits of labeled graphs are shown in \cref{fig:33graphslabeled}.

\begin{figure}
    \centering
    \includegraphics[width=0.9\linewidth]{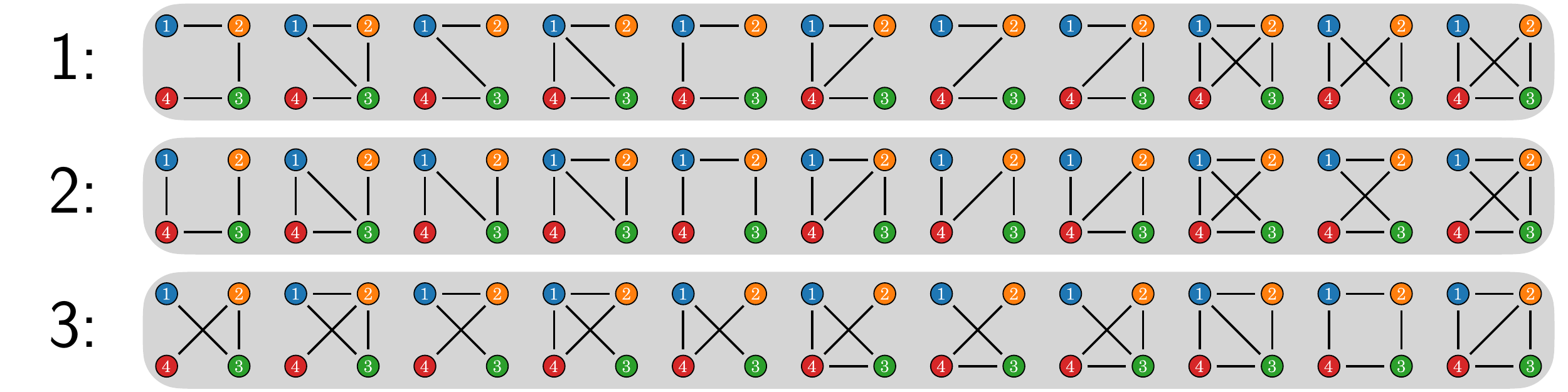}
    \caption{
    The three separate LU-orbits associated with the four-qubit labeled linear cluster state. Each row depicts a distinct LU-orbit, which is the collection of graph states that can be obtained by local unitary operations. Note that this includes the graph states that are associated with those graphs that are obtained after a permutation of the nodes of the original graph.
    Not all permutations lead to graph states that are in the same LU-orbit.
    }
    \label{fig:33graphslabeled}
\end{figure}

When the graphs are considered as unlabeled, there are effectively only four distinct graphs left out of the 33. These four graphs together form the LU-orbit of the unlabeled linear cluster state, and are shown on the left in \cref{fig:33graphsunlabeled}. With each of these graphs there are a selection of the other 29 graphs associated; those that are associated with the same graph are grouped together in the same figure.

\begin{figure}
    \centering
    \includegraphics[width=0.9\linewidth]{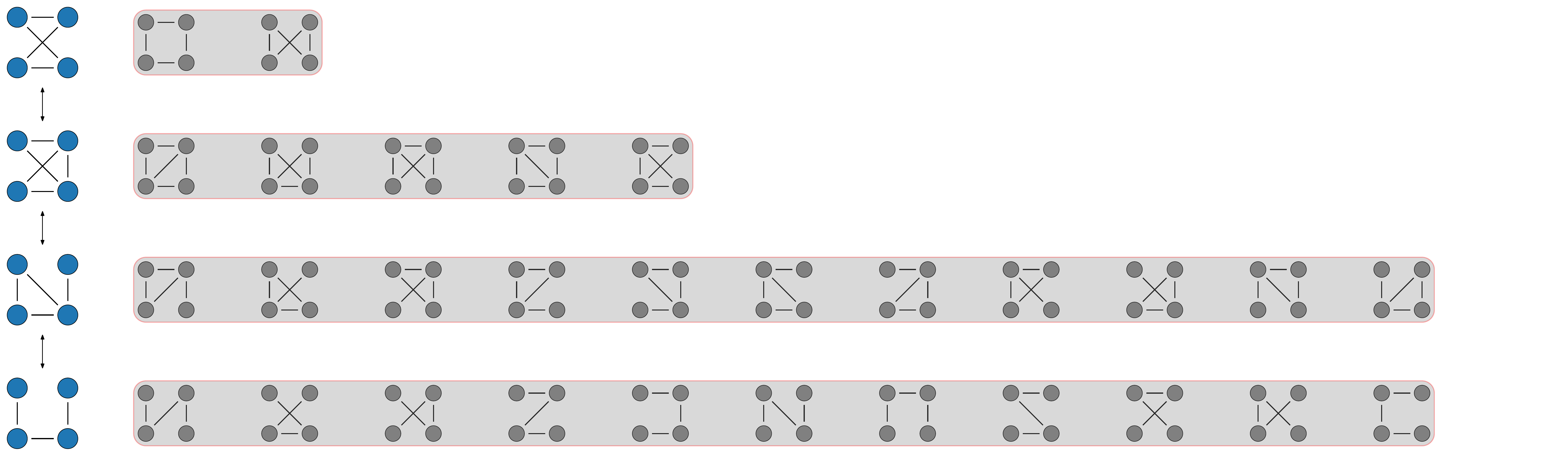}
    \caption{
    The orbit of the unlabeled line graph consists of only four graphs, shown in color on the left. Each of these unlabeled graphs has several other graphs associated with it, which are the permutations of the unlabeled graph and are therefore counted as the same. Note that each of the 33 labeled graphs from \cref{fig:33graphslabeled} can be identified with either one of the four unlabeled colored graphs on the left, or one of the 29 permutations in the gray boxes.}
    \label{fig:33graphsunlabeled}
\end{figure}

\section{Proof of Lemma 11} \label{app:Lemma2_proof}
In this section we prove \cref{Lemma:LHS_vs_RHS}, which establishes the following relation between the stabilizer set of a marginal state $\rho_M$ and its sub marginals:
    \begin{align}
    \label{eq:LHS_vs_RHS}
        \vert \SSS_M \vert = 2^\ell \vert \langle \bigcup_{N \subsetneq M} \SSS_N \rangle \vert,
    \end{align}
    where $\ell \in \lbrace 0, 1, 2 \rbrace$. Furthermore, if \cref{eq:LHS_vs_RHS} is satisfied for $\ell=2$, then the size of $M$ is an even number and  we have $\abs{\SSS_{M}} = 4$.

    \begin{figure}
        \centering
        \includegraphics[width = 0.3\textwidth]{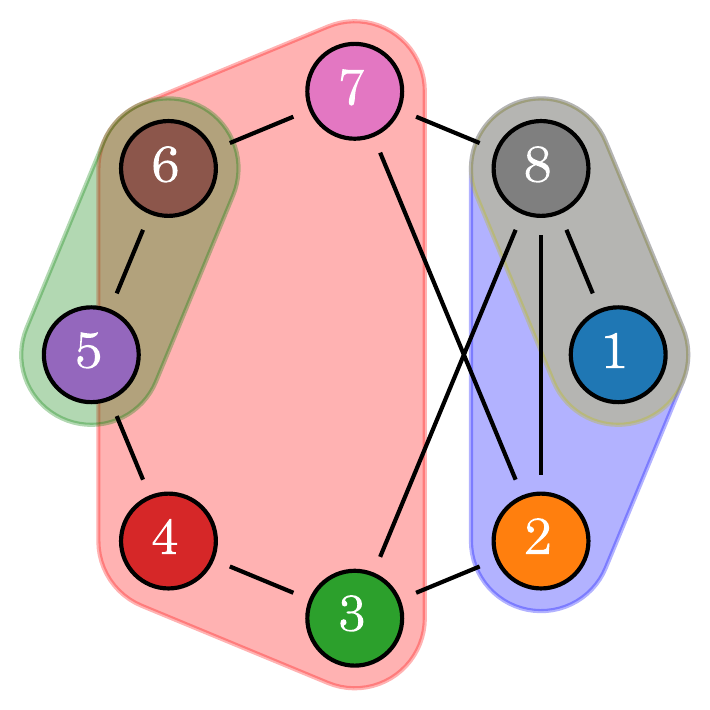} 
        \caption{A graph with subsets $M$ such that $\ell$ in \cref{eq:LHS_vs_RHS} takes values 0, 1, and 2. We can find many sets $M_0$ for which $\ell = 0$, for example $M_0 = \{5,6 \}$ for which both sides of \cref{eq:LHS_vs_RHS} are equal to 1. The set $M_{1}^{\text{even}} = \{1,8\}$ has a stabilizer of 2 elements. All proper subsets of $M_{1}^{\text{even}}$ have only one stabilizer element which is the identity  and therefore we get $\ell = 1$. For $M_{1}^{\text{odd}} = \{1,2,8\}$ we get $\ell = 1$ as well since the stabilizers $ \SSS_{\{ 1,2,8\}} = \{ \1, X_1 Z_8, - Z_1 Y_2 Y_8, - Y_1 Y_2 X_8   \}$, $ \SSS_{\{ 1, 8\}} = \{ \1, X_1 Z_8   \}$ and $ \SSS_N =\{ \1 \}$ for all other proper subsets $N \subsetneq M_{1}^{\text{odd}}$. An example where $\ell = 2$ is shown on the marginal set $M_2 = \{ 3,4,6,7 \}$. We have $\SSS_{\{ 3,4,6,7\}} = \{\1, X_3 Z_4 Z_6 X_7, Z_3 X_4 X_6 Z_7, Y_3 Y_4 Y_6 Y_7 \}$ and $\SSS_N = \{ \1 \}$ for all subsets $N \subsetneq M_2$.}
        \label{fig:example_interestingmarginals}
    \end{figure}

  \begin{proof}[Proof of \cref{Lemma:LHS_vs_RHS}]
     Notice that both sets $\langle \bigcup_{N \subsetneq M} \SSS_N \rangle $ and $\SSS_M $ are subgroups of the Pauli group and hence their order is a power of two. 
     Moreover, $\langle \bigcup_{N \subsetneq M} \SSS_N \rangle $ is a subgroup of $\SSS_M $, hence $\vert \SSS_M \vert = 2^\ell \vert \langle \bigcup_{N \subsetneq M} \SSS_N \rangle \vert$ for some natural number $\ell\in \mathbb{N}_0$. 

     We show in \cref{fig:example_interestingmarginals} by example that $\ell$ can take all values in $\lbrace 0, 1, 2 \rbrace$. We first show that $\ell \geqslant 3$ is not possible in general. We continue by proving  conditions on $M$ and $\SSS_M$ for $\ell = 2$. We proof all statements by contradiction.

     We now show that  $\ell \geqslant 3$ is impossible in general.  First, assume $\ell \geqslant 1$. Therefore, there is a Pauli string $\sigma_{1} \in \SSS_{M}$ for which $\sigma_{1} \not \in \langle \bigcup_{N \subsetneq M} \SSS_N \rangle$ and thus $\sigma_{1}$ has full support on the marginal $M$. 
     The set $\sigma_{1} \langle \bigcup_{N \subsetneq M} \SSS_N \rangle $ forms a coset of the subgroup $\langle \bigcup_{N \subsetneq M} \SSS_N \rangle  \in \SSS_{M}$.
     Now assume that $\ell \geqslant 2$. This implies the existence of some Pauli string $\sigma_{2} \in \SSS_{M}$ that is neither contained in $\langle \bigcup_{N \subsetneq M} \SSS_N \rangle $ nor in $\sigma_{1}\langle \bigcup_{N \subsetneq M} \SSS_N \rangle$ (i.e. the coset of $\sigma_{1}$).
     In particular, this implies the following: First, let $\sigma_{1} = \bigotimes_{i \in V} P_{i}$ and $\sigma_{2} = \bigotimes_{i \in V} Q_{i}$ for $P_{i}, Q_{i} \in \{X, Y, Z\}$. The above implies that $Q_{i} \not = P_{i}$ for every $i \in V$, for otherwise we would have $\sigma_{1} \sigma_{2} \in \SSS_N$ for some $N \subsetneq M$, which is a contradiction.

     Finally, assume $\ell \geqslant 3$. By the same reasoning as before, this implies the existence of a Pauli string $\sigma_{3} \in \SSS_{M}$ that is not in the sets $\langle \bigcup_{N \subsetneq M} \SSS_N \rangle  \in \SSS_{M}$, $\sigma_{1} \langle \bigcup_{N \subsetneq M} \SSS_N \rangle  \in \SSS_{M}$, or $\sigma_{2}\langle \bigcup_{N \subsetneq M} \SSS_N \rangle  \in \SSS_{M}$. 
     Similarly as before, $\sigma_{3}$ has full support on $M$ and has to be different from $\sigma_1$ and $\sigma_2$ on every position. Because of product relations of Pauli matrices, it follows that $\sigma_{3} = \pm \sigma_{1}\sigma_{2}$,  which is in contradiction with the assumption that $\sigma_{1}, \sigma_{2}$, and $\sigma_{3}$ are independent. Therefore $\ell =3$ is not possible.

     We continue by showing that if $\abs{M}$ is odd, we cannot have $\ell = 2$.
     We see that because of commutation relations of the Paulis, if $\abs{M}$ is odd, we have $\sigma_{1}\sigma_{2} = - \sigma_{2}\sigma_{1}$  which contradicts the Abelian structure of $\SSS_{M}$. Therefore for odd sets $M$, we cannot find different stabilizers $\sigma_1$ and  $\sigma_2$ with full support on $M$ and therefore $\ell=2$ is not possible.

    We finally show that if $\ell = 2$, we have $\langle \bigcup_{N \subsetneq M} \SSS_N \rangle = \{ \1 \}$.
    Suppose that there exists an element $\tau \in \langle \bigcup_{N \subsetneq M} \SSS_N \rangle$ where $\tau \not = I$. As stated before, when $\ell = 2$ there exist elements $\sigma_{1}, \sigma_{2}$ and $\sigma_{3} = \sigma_{1} \sigma_{2}$ that have 
    all different Pauli operators on all nodes. Since $\tau \not = I$, there is at least one node on which $\tau$ has non-trivial support. But then, at least $\tau \sigma_{1}$, $\tau \sigma_{2}$, or $\tau \sigma_{3}$ does not have full support; w.l.o.g. assume this is so for $\tau \sigma_{1}$. However, then $\tau$ and $\tau \sigma_{1}$ are both in $\langle \bigcup_{N \subsetneq M} \SSS_N \rangle$, which is in direct contradiction with the fact that $\sigma_{1} = \tau^2 \sigma_{1} \not \in \langle \bigcup_{N \subsetneq M} \SSS_N \rangle$. Therefore, no such $\tau$ can exist and it must be that $\langle \bigcup_{N \subsetneq M} \SSS_N \rangle = \{ \1 \}$.

\end{proof}

\section{Condensation Rules} \label{sec:proof_condens}

\subsection{Proof of Lemma $15$}

In this section, we prove \cref{lem:condens_big} which states that given two graphs $G$ and $G'$ and a condensation set $C$ such that $\stabdim{C} = \vert C \vert - 1$ then,  if 
    $G$ and $G'$ are LC-equivalent, it follows that $G_c$ and $G'_c$ are LC-equivalent.
We first introduce notation. 

Consider a graph $G = (V,E)$ and a set $C \subsetneq V$ such that $\stabdim{C} = \vert C \vert - 1$. We earlier defined the neighborhood of a set $\NN{C}$ as the set of nodes adjacent to nodes in $C$. 
In \cref{def:metagraph}, we defined metagraphs. For a graph $G$ and a set $C$, we defined the metagraph  $G_C$ which has two types of nodes. Type-1 nodes are those from set $C$, while type-2 nodes represent some nodes outside of $C$ and tell whether they are connected to a certain subset $B$ of $C$.   Here, we 
are interested in the set of nodes in $V \setminus C$ which get represented by type-2 nodes. We denote this set  $\NNh{B}$ and define it such that all nodes in $B \in  \mathcal{P}(C) \setminus \emptyset $ are adjacent with all nodes in  $\NNh{B}$ and $B$ is the largest subset of $C$ for which this is true:
\begin{align}
    \NNh{B}
    \coloneqq \left\{ v \in V \setminus C \, \middle| \, \begin{array}{l}
    \forall b \in B: \, (v,b) \in E \text{ and} \\
    \forall c \in C \setminus B: \, (v,c) \notin E
    \end{array} \right\}.
\end{align}
We call such subset of the complete neighborhood a neighborset of $B$.
By definition, we have
\begin{align}
    \NN{C} = \dot \bigcup_{B \in \mathcal{P}(C) \setminus \emptyset} \NNh{B}.
\end{align}
Note that for most $B \subseteq C$, $\NNh{B} = \emptyset$. We count how many sets $\NNh{B}$ are non-empty and call this number the amount of non-empty neighborhoods of $C$.

The proof is divided into two parts. 
(i) We first show that every  set $C$ with $\stabdim{C} = \vert C \vert - 1$ has exactly one non-empty neighborhood. That is, there is a $B \subseteq C$, $B \neq \emptyset$ such that $\NN{C} = \NNh{B}$ if and only if $\stabdim{C} = \vert C \vert - 1$.
(ii) We then show that for two graphs which are LC-equivalent also the condensed graphs are LC-equivalent. 

\begin{proof}
(i) \underline{$\exists B \subseteq C$, $B \neq \emptyset$: $\NN{C} = \NNh{B} \Leftrightarrow \stabdim{C} = \vert C \vert - 1$}.

\underline{$\Rightarrow$} Consider a set $C \subseteq V$ with a neighborhood $\NN{C}$ and a set $B \subseteq C$ with $B \neq \emptyset$ such that $\NNh{B} = \NN{C}$. We count the number of elements in the stabilizer set $\SSS_C$ and compute $\stabdim{C}$ from this number. The stabilizer set is a subset of the set generated by all generators $g_i$ for $i\in C$:
\begin{align}
    \SSS_C \subseteq \langle \lbrace g_i \mid i \in C \rbrace \rangle, 
\end{align}
and therefore the number of elements in $\SSS_C$ is upper bounded by $2^{\vert C \vert}$.

Consider the case where $B = \lbrace b \rbrace$, that is, $B$ contains only one element. We have that $\tr_{V\setminus B}(g_c) = g_c$ for all $c \in C \setminus B$. Therefore $\langle g_c \mid c \in C \setminus B \rangle \subseteq \SSS_C$ and therefore $\SSS_C$ has at least $2^{\vert C \vert - 1}$ elements.   Further  we have $\tr_{V\setminus B}(g_b) = 0$ and therefore $g_b \notin \SSS_C$. Also all products containing $g_b$ are not in $\SSS_b$, which are $2^{\vert C \vert - 1}$ terms. Therefore $\SSS_C$ has at most $2^{\vert C \vert} - 2^{\vert C \vert - 1} = 2^{\vert C \vert - 1}$ elements. It follows that $\SSS_C$ has exactly $2^{\vert C \vert - 1}$ elements, which corresponds to $\stabdim{C} = \vert C \vert - 1$.

For $\vert B \vert > 1$ a similar argumentation can be made. We have that $\tr_{V\setminus B}(g_c) = g_c$ for all $c \in C \setminus B$, which are $\vert C \vert - \vert B \vert$ elements. Also $g_b g_{b'} \in \SSS_C$ for all $b, b' \in B$, $b \neq b'$, which are $\vert B \vert - 1$ (algebraically) independent elements.  Therefore $\SSS_C$ has at least $2^{(\vert C \vert - \vert B \vert) + (\vert B \vert - 1)} = 2^{\vert C \vert - 1}$ elements. Further, all products of generators which contain an odd number of generators $g_b$, $b \in B$ are not in $\SSS_C$. This are  $2^{\vert C \vert - 1}$ elements. Therefore $\SSS_C$ contains exactly $2^{\vert C \vert - 1}$ elements, which corresponds to $\stabdim{C} = \vert C \vert - 1$.

\underline{$\Leftarrow$} Proof by contradiction. Assume that $C$ has more than one non-empty neighborhood. That is there are two sets $B_1, B_2 \in C$, both not empty and  $B_1 \neq B_2$ which have a non-empty neighborhood:  $ \NNh{B_1} \dot\cup \NNh{B_2} \subseteq \NN{C}$. This is the case if $C$ has at least two non-empty neighborhoods. There are three relevant relations between $B_1$ and $B_2$. We show that in all cases $\stabdim{C} < \vert C \vert - 1$.

\begin{enumerate}[a)]
    \item \underline{$B_1 \cap B_2 = \emptyset$}: All products of stabilizer elements which contain odd numbers of generators $g_b$, $b\in B_1$ are not in $\SSS_C$, which are $2^{\vert C \vert - 1}$ elements. There is at least one more product of generators which contains an odd number of generators $g_b$, $b \in B_2$. Therefore $\vert \SSS_C \vert \leqslant 2^{\vert C \vert - 1} - 1$ and therefore $\stabdim{C} < \vert C \vert - 1$.
    \item \underline{$B_1 \subsetneq B_2$}: ($B_2 \subsetneq B_1$ equivalent.) Same argumentation as in (a) for the disjoint sets $B_1$ and $B_2 \setminus B_1$.
    \item \underline{$B_1 \cap B_2 \neq \emptyset, B_1 \nsubseteq B_2, B_2 \nsubseteq B_1$}: Same argumentation as in (a) for the disjoint sets $B_1 \setminus B_2$ and $B_2 \setminus B_1$.
\end{enumerate}

Therefore, if $C$ has more than one non-empty neighborhood, we have $\stabdim{C} < \vert C \vert - 1$.

It follows that it is equivalent that $C$ has exactly one non-empty neighborhood and $\stabdim{C} = \vert C \vert - 1$.

(ii) We show that for two graphs $G$ and $G'$ which are LC-equivalent the condensed graphs $G_C$ and $G'_C$ are also equivalent where $C$ is chosen such that $\stabdim{C} = \vert C \vert - 1$. Without loss of generality we assume that there is a $j \in V$ such that $G' = \LC_j (G)$, otherwise the argumentation has to be applied multiple times.

We divide the node set $V$ into four disjoint sets:
\begin{equation}
V = B \mathbin{\dot{\cup}} (C \setminus B) \mathbin{\dot{\cup}} \NN{C} \mathbin{\dot{\cup}} (V \setminus (C \cup \NN{C})).
\end{equation}

As shown in (i), for the chosen set $C$ we have $\NN{C} = \NNh{B}$.
We show that for local complementation on any node in one of the sets, the statement is true.

\begin{enumerate}[a)]
    \item \underline{$j \in V \setminus ( C \cup \NN{C} )$}: Local complementation on a node which is not adjacent to $C$, does not change the structure of the graph in the set $C$ and therefore has the same effect on a graph before and after condensation. That is $G'_C = \LC_j (G_C)$.
    \item \underline{$j \in C\setminus B$}: Local complementation can only change the connectivity structure in $C$, what we do not see after condensation. That is $G'_C = G_C$.
    \item \underline{$j \in B$}: Local complementation can change the edge structure of $C$, of $\NN{C}$ as well as the set $B$ to another set $B' \neq \emptyset$. Changes in $C$ have no effect. Changes in $\NN{C}$ are the same as local complementation on node $c$ on $G_C$. The change of $B$ is such that $\NNh{B} = \NNh{B'}$, that is $C$ still has exactly one non-empty neighborhood which is the same as before. That is $G'_C = \LC_c (G_C)$.
    \item \underline{$j \in \NN{C}$}: Local complementation can change the edge structure in $C$, which does not matter for condensation. It further can change the neighborhood of $C$ in the same way as it changes the neighborhood of $c$ in the condensed graph. It can also change other edges which are not affected by condensation in $C$. That is $G'_C = \LC_j (G_C)$. 
\end{enumerate}
Therefore in all cases $G_C$ and $G'_C$ are LC-equivalent after condensation.

\end{proof}

\subsection{Not Working Condensation Rules}
Below we give a list of condensation rules which are not sufficient to help deciding LC-equivalence of the initial graphs.  
\begin{Observation} \label{obs:not-condense}
    We give a list of condensation rules which preserve LC-equivalence in some cases but do not hold in full generality.
    \begin{enumerate}
        \item It is not correct that a set $C$ can be condensed whenever $\vert \SSS_C \vert > \vert \langle \bigcup_{B \subsetneq C} \SSS_B \rangle \vert$ is fulfilled.
        \item A condensation set $C$ with a neighborhood which can get composed into neighborsets which are not adjacent by a path outside $C$ cannot be condensed in general. 
        \item It is not true that a set $C$ such that the neighborhood can get composed into non-empty neighborsets $\NNh{D}$ for some $D \in \DD$ where all $D \in \DD$ are pairwise disjoint can be condensed in general. $\DD$ is the set of $D \subseteq C$ for which the neighborset is not empty. 
    \end{enumerate}
\end{Observation}

\begin{figure}
    \centering
    (a)
    \includegraphics[width=.4\linewidth]{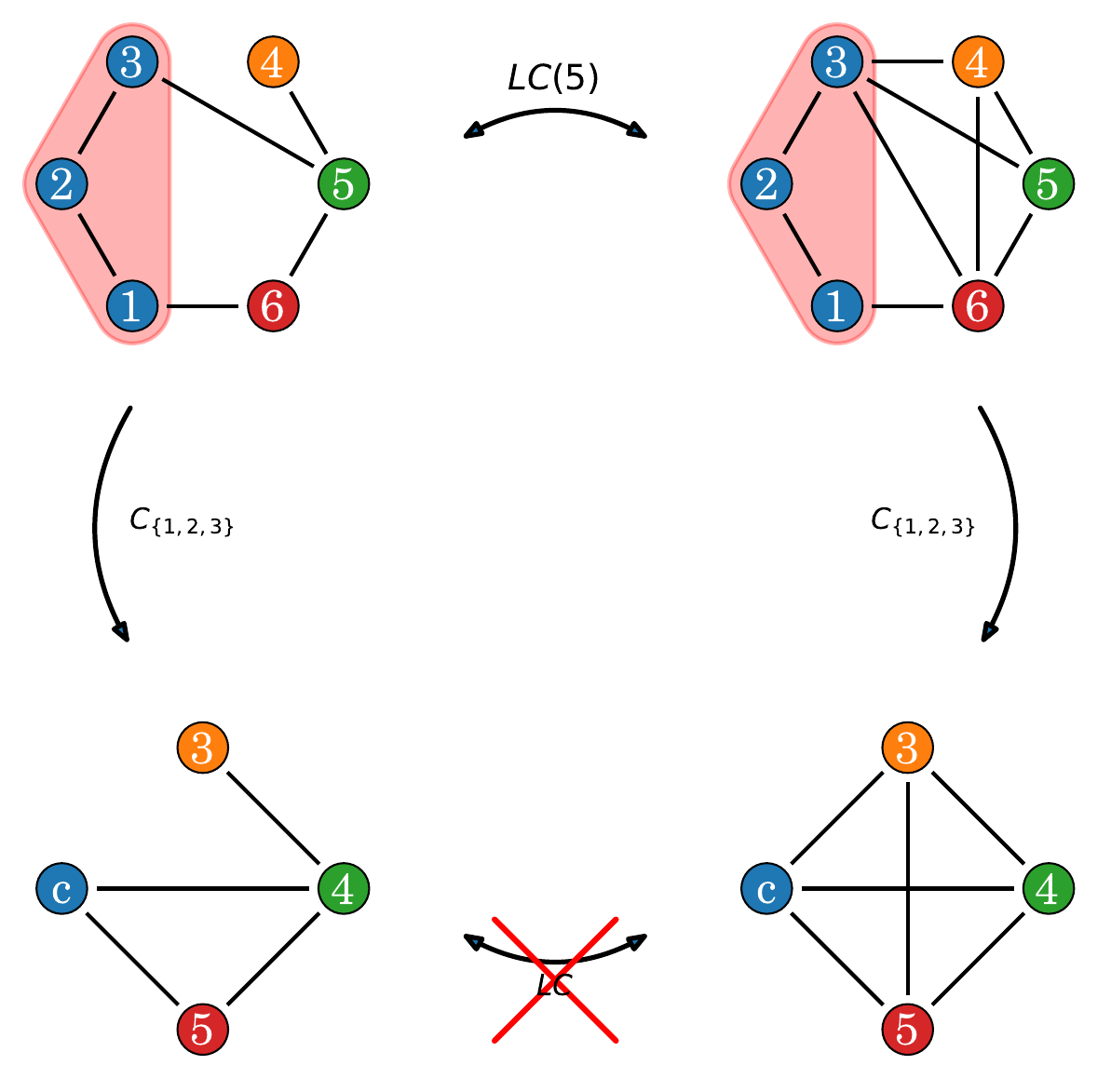}
    (b)
    \includegraphics[width=.4\linewidth]{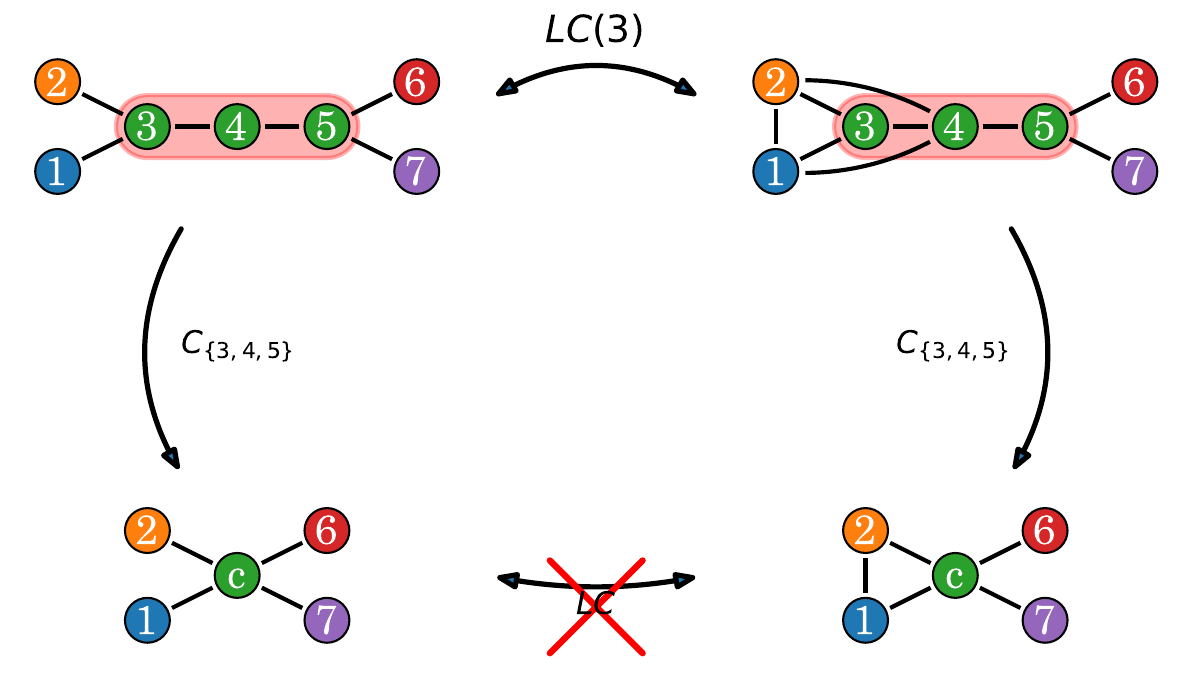}
    \caption{Counterexamples for not working condensation rules, presented in \cref{obs:not-condense}. In both (a) and (b) we chose a condensation set $C$ such that $\stabdim{C} = 1$ and \cref{eq:threecondset} is fulfilled. In (b) we have additionally that the neighborsets $\NNh{3} = \lbrace 1,2 \rbrace$, $\NNh{5} = \lbrace 6,7 \rbrace$ are not adjacent by a path outside $C$ and that $\lbrace 3 \rbrace$ and $\lbrace 5 \rbrace$ are disjoint.}
    \label{fig:counterex_condense}
\end{figure}

Counterexamples are shown in \cref{fig:counterex_condense}.

\bibliographystyle{apsrev4-2}
\bibliography{main}

\begin{thebibliography}{44}%
\makeatletter
\providecommand \@ifxundefined [1]{%
 \@ifx{#1\undefined}
}%
\providecommand \@ifnum [1]{%
 \ifnum #1\expandafter \@firstoftwo
 \else \expandafter \@secondoftwo
 \fi
}%
\providecommand \@ifx [1]{%
 \ifx #1\expandafter \@firstoftwo
 \else \expandafter \@secondoftwo
 \fi
}%
\providecommand \natexlab [1]{#1}%
\providecommand \enquote  [1]{``#1''}%
\providecommand \bibnamefont  [1]{#1}%
\providecommand \bibfnamefont [1]{#1}%
\providecommand \citenamefont [1]{#1}%
\providecommand \href@noop [0]{\@secondoftwo}%
\providecommand \href [0]{\begingroup \@sanitize@url \@href}%
\providecommand \@href[1]{\@@startlink{#1}\@@href}%
\providecommand \@@href[1]{\endgroup#1\@@endlink}%
\providecommand \@sanitize@url [0]{\catcode `\\12\catcode `\$12\catcode
  `\&12\catcode `\#12\catcode `\^12\catcode `\_12\catcode `\%12\relax}%
\providecommand \@@startlink[1]{}%
\providecommand \@@endlink[0]{}%
\providecommand \url  [0]{\begingroup\@sanitize@url \@url }%
\providecommand \@url [1]{\endgroup\@href {#1}{\urlprefix }}%
\providecommand \urlprefix  [0]{URL }%
\providecommand \Eprint [0]{\href }%
\providecommand \doibase [0]{https://doi.org/}%
\providecommand \selectlanguage [0]{\@gobble}%
\providecommand \bibinfo  [0]{\@secondoftwo}%
\providecommand \bibfield  [0]{\@secondoftwo}%
\providecommand \translation [1]{[#1]}%
\providecommand \BibitemOpen [0]{}%
\providecommand \bibitemStop [0]{}%
\providecommand \bibitemNoStop [0]{.\EOS\space}%
\providecommand \EOS [0]{\spacefactor3000\relax}%
\providecommand \BibitemShut  [1]{\csname bibitem#1\endcsname}%
\let\auto@bib@innerbib\@empty
\bibitem [{\citenamefont {Murta}\ \emph {et~al.}(2020)\citenamefont {Murta},
  \citenamefont {Grasselli}, \citenamefont {Kampermann},\ and\ \citenamefont
  {Bruß}}]{Murta_2020_quantum_conference}%
  \BibitemOpen
  \bibfield  {author} {\bibinfo {author} {\bibfnamefont {G.}~\bibnamefont
  {Murta}}, \bibinfo {author} {\bibfnamefont {F.}~\bibnamefont {Grasselli}},
  \bibinfo {author} {\bibfnamefont {H.}~\bibnamefont {Kampermann}},\ and\
  \bibinfo {author} {\bibfnamefont {D.}~\bibnamefont {Bruß}},\ }\href
  {http://dx.doi.org/10.1002/qute.202000025} {\bibfield  {journal} {\bibinfo
  {journal} {Adv. Quantum Technol.}\ }\textbf {\bibinfo {volume} {3}},\
  \bibinfo {pages} {2000025} (\bibinfo {year} {2020})}\BibitemShut {NoStop}%
\bibitem [{\citenamefont {Raussendorf}\ \emph {et~al.}(2003)\citenamefont
  {Raussendorf}, \citenamefont {Browne},\ and\ \citenamefont
  {Briegel}}]{Raussendorf2003measurement}%
  \BibitemOpen
  \bibfield  {author} {\bibinfo {author} {\bibfnamefont {R.}~\bibnamefont
  {Raussendorf}}, \bibinfo {author} {\bibfnamefont {D.~E.}\ \bibnamefont
  {Browne}},\ and\ \bibinfo {author} {\bibfnamefont {H.~J.}\ \bibnamefont
  {Briegel}},\ }\href {https://doi.org/10.1103/PhysRevA.68.022312} {\bibfield
  {journal} {\bibinfo  {journal} {Phys. Rev. A}\ }\textbf {\bibinfo {volume}
  {68}},\ \bibinfo {pages} {022312} (\bibinfo {year} {2003})}\BibitemShut
  {NoStop}%
\bibitem [{\citenamefont {Terhal}(2015)}]{Terhal2015errorcorrection}%
  \BibitemOpen
  \bibfield  {author} {\bibinfo {author} {\bibfnamefont {B.~M.}\ \bibnamefont
  {Terhal}},\ }\href {https://doi.org/10.1103/RevModPhys.87.307} {\bibfield
  {journal} {\bibinfo  {journal} {Rev. Mod. Phys.}\ }\textbf {\bibinfo {volume}
  {87}},\ \bibinfo {pages} {307} (\bibinfo {year} {2015})}\BibitemShut
  {NoStop}%
\bibitem [{\citenamefont {Horodecki}\ \emph {et~al.}(2009)\citenamefont
  {Horodecki}, \citenamefont {Horodecki}, \citenamefont {Horodecki},\ and\
  \citenamefont {Horodecki}}]{Horodecki2009entanglement}%
  \BibitemOpen
  \bibfield  {author} {\bibinfo {author} {\bibfnamefont {R.}~\bibnamefont
  {Horodecki}}, \bibinfo {author} {\bibfnamefont {P.}~\bibnamefont
  {Horodecki}}, \bibinfo {author} {\bibfnamefont {M.}~\bibnamefont
  {Horodecki}},\ and\ \bibinfo {author} {\bibfnamefont {K.}~\bibnamefont
  {Horodecki}},\ }\href {https://doi.org/10.1103/RevModPhys.81.865} {\bibfield
  {journal} {\bibinfo  {journal} {Rev. Mod. Phys.}\ }\textbf {\bibinfo {volume}
  {81}},\ \bibinfo {pages} {865} (\bibinfo {year} {2009})}\BibitemShut
  {NoStop}%
\bibitem [{\citenamefont {Gühne}\ and\ \citenamefont
  {Tóth}(2009)}]{Guehne2009entanglement}%
  \BibitemOpen
  \bibfield  {author} {\bibinfo {author} {\bibfnamefont {O.}~\bibnamefont
  {Gühne}}\ and\ \bibinfo {author} {\bibfnamefont {G.}~\bibnamefont {Tóth}},\
  }\href {https://doi.org/https://doi.org/10.1016/j.physrep.2009.02.004}
  {\bibfield  {journal} {\bibinfo  {journal} {Physics Reports}\ }\textbf
  {\bibinfo {volume} {474}},\ \bibinfo {pages} {1} (\bibinfo {year}
  {2009})}\BibitemShut {NoStop}%
\bibitem [{\citenamefont {Hein}\ \emph {et~al.}(2006)\citenamefont {Hein},
  \citenamefont {Dür}, \citenamefont {Eisert}, \citenamefont {Raussendorf},
  \citenamefont {den Nest},\ and\ \citenamefont
  {Briegel}}]{hein2006entanglement}%
  \BibitemOpen
  \bibfield  {author} {\bibinfo {author} {\bibfnamefont {M.}~\bibnamefont
  {Hein}}, \bibinfo {author} {\bibfnamefont {W.}~\bibnamefont {Dür}}, \bibinfo
  {author} {\bibfnamefont {J.}~\bibnamefont {Eisert}}, \bibinfo {author}
  {\bibfnamefont {R.}~\bibnamefont {Raussendorf}}, \bibinfo {author}
  {\bibfnamefont {M.~V.}\ \bibnamefont {den Nest}},\ and\ \bibinfo {author}
  {\bibfnamefont {H.~J.}\ \bibnamefont {Briegel}},\ }\href
  {https://doi.org/10.3254/978-1-61499-018-5-115} {\bibfield  {journal}
  {\bibinfo  {journal} {Proceedings of the International School of Physics
  “Enrico Fermi”}\ }\textbf {\bibinfo {volume} {162}},\ \bibinfo {pages}
  {115} (\bibinfo {year} {2006})},\ \bibinfo {note}
  {ar{X}iv:quant-ph/0602096}\BibitemShut {NoStop}%
\bibitem [{\citenamefont {Majidy}\ \emph {et~al.}(2024)\citenamefont {Majidy},
  \citenamefont {Wilson},\ and\ \citenamefont
  {Laflamme}}]{Majidy_Wilson_Laflamme_2024}%
  \BibitemOpen
  \bibfield  {author} {\bibinfo {author} {\bibfnamefont {S.}~\bibnamefont
  {Majidy}}, \bibinfo {author} {\bibfnamefont {C.}~\bibnamefont {Wilson}},\
  and\ \bibinfo {author} {\bibfnamefont {R.}~\bibnamefont {Laflamme}},\ }\href
  {https://doi.org/10.1017/9781009417020} {\emph {\bibinfo {title} {Building
  Quantum Computers: A Practical Introduction}}}\ (\bibinfo  {publisher}
  {Cambridge University Press},\ \bibinfo {year} {2024})\BibitemShut {NoStop}%
\bibitem [{\citenamefont {Kraus}(2010{\natexlab{a}})}]{Kraus2010_PRL}%
  \BibitemOpen
  \bibfield  {author} {\bibinfo {author} {\bibfnamefont {B.}~\bibnamefont
  {Kraus}},\ }\href {https://doi.org/10.1103/PhysRevLett.104.020504} {\bibfield
   {journal} {\bibinfo  {journal} {Phys. Rev. Lett.}\ }\textbf {\bibinfo
  {volume} {104}},\ \bibinfo {pages} {020504} (\bibinfo {year}
  {2010}{\natexlab{a}})}\BibitemShut {NoStop}%
\bibitem [{\citenamefont {Maciążek}\ \emph {et~al.}(2013)\citenamefont
  {Maciążek}, \citenamefont {Oszmaniec},\ and\ \citenamefont
  {Sawicki}}]{Maciazek2013howmany}%
  \BibitemOpen
  \bibfield  {author} {\bibinfo {author} {\bibfnamefont {T.}~\bibnamefont
  {Maciążek}}, \bibinfo {author} {\bibfnamefont {M.}~\bibnamefont
  {Oszmaniec}},\ and\ \bibinfo {author} {\bibfnamefont {A.}~\bibnamefont
  {Sawicki}},\ }\href {https://doi.org/10.1063/1.4819499} {\bibfield  {journal}
  {\bibinfo  {journal} {J. Math. Phys.}\ }\textbf {\bibinfo {volume} {54}},\
  \bibinfo {pages} {092201} (\bibinfo {year} {2013})}\BibitemShut {NoStop}%
\bibitem [{\citenamefont {Van~den Nest}\ \emph
  {et~al.}(2004{\natexlab{a}})\citenamefont {Van~den Nest}, \citenamefont
  {Dehaene},\ and\ \citenamefont {De~Moor}}]{Van_den_Nest_2004_graphical}%
  \BibitemOpen
  \bibfield  {author} {\bibinfo {author} {\bibfnamefont {M.}~\bibnamefont
  {Van~den Nest}}, \bibinfo {author} {\bibfnamefont {J.}~\bibnamefont
  {Dehaene}},\ and\ \bibinfo {author} {\bibfnamefont {B.}~\bibnamefont
  {De~Moor}},\ }\href {https://doi.org/10.1103/physreva.69.022316} {\bibfield
  {journal} {\bibinfo  {journal} {Phys. Rev. A}\ }\textbf {\bibinfo {volume}
  {69}},\ \bibinfo {pages} {022316} (\bibinfo {year}
  {2004}{\natexlab{a}})}\BibitemShut {NoStop}%
\bibitem [{\citenamefont {Nest}\ \emph {et~al.}(2005)\citenamefont {Nest},
  \citenamefont {Dehaene},\ and\ \citenamefont
  {De~Moor}}]{nestLocalUnitaryLocal2005}%
  \BibitemOpen
  \bibfield  {author} {\bibinfo {author} {\bibfnamefont {M.~V.~d.}\
  \bibnamefont {Nest}}, \bibinfo {author} {\bibfnamefont {J.}~\bibnamefont
  {Dehaene}},\ and\ \bibinfo {author} {\bibfnamefont {B.}~\bibnamefont
  {De~Moor}},\ }\href {https://doi.org/10.1103/PhysRevA.71.062323} {\bibfield
  {journal} {\bibinfo  {journal} {Physical Review A}\ }\textbf {\bibinfo
  {volume} {71}},\ \bibinfo {pages} {062323} (\bibinfo {year}
  {2005})}\BibitemShut {NoStop}%
\bibitem [{\citenamefont {Ji}\ \emph {et~al.}(2010)\citenamefont {Ji},
  \citenamefont {Chen}, \citenamefont {Wei},\ and\ \citenamefont
  {Ying}}]{ji2010lulcconj}%
  \BibitemOpen
  \bibfield  {author} {\bibinfo {author} {\bibfnamefont {Z.}~\bibnamefont
  {Ji}}, \bibinfo {author} {\bibfnamefont {J.}~\bibnamefont {Chen}}, \bibinfo
  {author} {\bibfnamefont {Z.}~\bibnamefont {Wei}},\ and\ \bibinfo {author}
  {\bibfnamefont {M.}~\bibnamefont {Ying}},\ }\href
  {https://doi.org/10.26421/QIC10.1-2-8} {\bibfield  {journal} {\bibinfo
  {journal} {Quantum Inf. Comput.}\ }\textbf {\bibinfo {volume} {10}},\
  \bibinfo {pages} {97} (\bibinfo {year} {2010})}\BibitemShut {NoStop}%
\bibitem [{\citenamefont {Tsimakuridze}\ and\ \citenamefont
  {Gühne}(2017)}]{Tsimakuridze2017Graphstates}%
  \BibitemOpen
  \bibfield  {author} {\bibinfo {author} {\bibfnamefont {N.}~\bibnamefont
  {Tsimakuridze}}\ and\ \bibinfo {author} {\bibfnamefont {O.}~\bibnamefont
  {Gühne}},\ }\href {https://doi.org/10.1088/1751-8121/aa67cd} {\bibfield
  {journal} {\bibinfo  {journal} {J. Phys. A Math. Theor.}\ }\textbf {\bibinfo
  {volume} {50}},\ \bibinfo {pages} {195302} (\bibinfo {year}
  {2017})}\BibitemShut {NoStop}%
\bibitem [{\citenamefont {Cabello}\ \emph {et~al.}(2009)\citenamefont
  {Cabello}, \citenamefont {López-Tarrida}, \citenamefont {Moreno},\ and\
  \citenamefont {Portillo}}]{cabelloEntanglementEightqubitGraph2009}%
  \BibitemOpen
  \bibfield  {author} {\bibinfo {author} {\bibfnamefont {A.}~\bibnamefont
  {Cabello}}, \bibinfo {author} {\bibfnamefont {A.~J.}\ \bibnamefont
  {López-Tarrida}}, \bibinfo {author} {\bibfnamefont {P.}~\bibnamefont
  {Moreno}},\ and\ \bibinfo {author} {\bibfnamefont {J.~R.}\ \bibnamefont
  {Portillo}},\ }\href {https://doi.org/10.1016/j.physleta.2009.04.055}
  {\bibfield  {journal} {\bibinfo  {journal} {Phys. Lett. A}\ }\textbf
  {\bibinfo {volume} {373}},\ \bibinfo {pages} {2219} (\bibinfo {year}
  {2009})}\BibitemShut {NoStop}%
\bibitem [{\citenamefont {Dahlberg}\ \emph {et~al.}(2020)\citenamefont
  {Dahlberg}, \citenamefont {Helsen},\ and\ \citenamefont
  {Wehner}}]{Dahlberg2020Counting}%
  \BibitemOpen
  \bibfield  {author} {\bibinfo {author} {\bibfnamefont {A.}~\bibnamefont
  {Dahlberg}}, \bibinfo {author} {\bibfnamefont {J.}~\bibnamefont {Helsen}},\
  and\ \bibinfo {author} {\bibfnamefont {S.}~\bibnamefont {Wehner}},\ }\href
  {https://doi.org/10.1063/1.5120591} {\bibfield  {journal} {\bibinfo
  {journal} {J. Math. Phys.}\ }\textbf {\bibinfo {volume} {61}},\ \bibinfo
  {pages} {022202} (\bibinfo {year} {2020})}\BibitemShut {NoStop}%
\bibitem [{\citenamefont {Bouchet}(1993)}]{bouchet1993Recognizing}%
  \BibitemOpen
  \bibfield  {author} {\bibinfo {author} {\bibfnamefont {A.}~\bibnamefont
  {Bouchet}},\ }\href
  {https://doi.org/https://doi.org/10.1016/0012-365X(93)90357-Y} {\bibfield
  {journal} {\bibinfo  {journal} {Discrete Math.}\ }\textbf {\bibinfo {volume}
  {114}},\ \bibinfo {pages} {75} (\bibinfo {year} {1993})}\BibitemShut
  {NoStop}%
\bibitem [{\citenamefont {Van~den Nest}\ \emph
  {et~al.}(2004{\natexlab{b}})\citenamefont {Van~den Nest}, \citenamefont
  {Dehaene},\ and\ \citenamefont
  {De~Moor}}]{vandennestEfficientAlgorithmRecognize2004}%
  \BibitemOpen
  \bibfield  {author} {\bibinfo {author} {\bibfnamefont {M.}~\bibnamefont
  {Van~den Nest}}, \bibinfo {author} {\bibfnamefont {J.}~\bibnamefont
  {Dehaene}},\ and\ \bibinfo {author} {\bibfnamefont {B.}~\bibnamefont
  {De~Moor}},\ }\href {https://doi.org/10.1103/PhysRevA.70.034302} {\bibfield
  {journal} {\bibinfo  {journal} {Phys. Rev. A}\ }\textbf {\bibinfo {volume}
  {70}},\ \bibinfo {pages} {034302} (\bibinfo {year}
  {2004}{\natexlab{b}})}\BibitemShut {NoStop}%
\bibitem [{\citenamefont {Hein}\ \emph {et~al.}(2004)\citenamefont {Hein},
  \citenamefont {Eisert},\ and\ \citenamefont {Briegel}}]{Hein2004_7qubits}%
  \BibitemOpen
  \bibfield  {author} {\bibinfo {author} {\bibfnamefont {M.}~\bibnamefont
  {Hein}}, \bibinfo {author} {\bibfnamefont {J.}~\bibnamefont {Eisert}},\ and\
  \bibinfo {author} {\bibfnamefont {H.~J.}\ \bibnamefont {Briegel}},\ }\href
  {https://doi.org/10.1103/PhysRevA.69.062311} {\bibfield  {journal} {\bibinfo
  {journal} {Phys. Rev. A}\ }\textbf {\bibinfo {volume} {69}},\ \bibinfo
  {pages} {062311} (\bibinfo {year} {2004})}\BibitemShut {NoStop}%
\bibitem [{\citenamefont {Hajdušek}\ and\ \citenamefont
  {Murao}(2013)}]{Hajdusek_2013}%
  \BibitemOpen
  \bibfield  {author} {\bibinfo {author} {\bibfnamefont {M.}~\bibnamefont
  {Hajdušek}}\ and\ \bibinfo {author} {\bibfnamefont {M.}~\bibnamefont
  {Murao}},\ }\href {https://doi.org/10.1088/1367-2630/15/1/013039} {\bibfield
  {journal} {\bibinfo  {journal} {New J. Phys.}\ }\textbf {\bibinfo {volume}
  {15}},\ \bibinfo {pages} {013039} (\bibinfo {year} {2013})}\BibitemShut
  {NoStop}%
\bibitem [{\citenamefont {Burchardt}\ and\ \citenamefont
  {Raissi}(2020)}]{PhysRevA.102.022413}%
  \BibitemOpen
  \bibfield  {author} {\bibinfo {author} {\bibfnamefont {A.}~\bibnamefont
  {Burchardt}}\ and\ \bibinfo {author} {\bibfnamefont {Z.}~\bibnamefont
  {Raissi}},\ }\href {https://doi.org/10.1103/PhysRevA.102.022413} {\bibfield
  {journal} {\bibinfo  {journal} {Phys. Rev. A}\ }\textbf {\bibinfo {volume}
  {102}},\ \bibinfo {pages} {022413} (\bibinfo {year} {2020})}\BibitemShut
  {NoStop}%
\bibitem [{\citenamefont {Raissi}\ \emph {et~al.}(2022)\citenamefont {Raissi},
  \citenamefont {Burchardt},\ and\ \citenamefont
  {Barnes}}]{PhysRevA.106.062424}%
  \BibitemOpen
  \bibfield  {author} {\bibinfo {author} {\bibfnamefont {Z.}~\bibnamefont
  {Raissi}}, \bibinfo {author} {\bibfnamefont {A.}~\bibnamefont {Burchardt}},\
  and\ \bibinfo {author} {\bibfnamefont {E.}~\bibnamefont {Barnes}},\ }\href
  {https://doi.org/10.1103/PhysRevA.106.062424} {\bibfield  {journal} {\bibinfo
   {journal} {Phys. Rev. A}\ }\textbf {\bibinfo {volume} {106}},\ \bibinfo
  {pages} {062424} (\bibinfo {year} {2022})}\BibitemShut {NoStop}%
\bibitem [{\citenamefont {Burchardt}\ and\ \citenamefont
  {Hahn}(2025)}]{burchardt2023foliage}%
  \BibitemOpen
  \bibfield  {author} {\bibinfo {author} {\bibfnamefont {A.}~\bibnamefont
  {Burchardt}}\ and\ \bibinfo {author} {\bibfnamefont {F.}~\bibnamefont
  {Hahn}},\ }\href {https://doi.org/10.22331/q-2025-04-24-1720} {\bibfield
  {journal} {\bibinfo  {journal} {Quantum}\ }\textbf {\bibinfo {volume} {9}},\
  \bibinfo {pages} {1720} (\bibinfo {year} {2025})}\BibitemShut {NoStop}%
\bibitem [{\citenamefont {Greenberger}\ \emph {et~al.}(1989)\citenamefont
  {Greenberger}, \citenamefont {Horne},\ and\ \citenamefont {Zeilinger}}]{GHZ}%
  \BibitemOpen
  \bibfield  {author} {\bibinfo {author} {\bibfnamefont {D.~M.}\ \bibnamefont
  {Greenberger}}, \bibinfo {author} {\bibfnamefont {M.~A.}\ \bibnamefont
  {Horne}},\ and\ \bibinfo {author} {\bibfnamefont {A.}~\bibnamefont
  {Zeilinger}},\ }\bibinfo {title} {Going beyond {B}ell's theorem},\ in\ \href
  {https://doi.org/10.1007/978-94-017-0849-4_10} {\emph {\bibinfo {booktitle}
  {Bell's Theorem, Quantum Theory and Conceptions of the Universe}}},\ \bibinfo
  {editor} {edited by\ \bibinfo {editor} {\bibfnamefont {M.}~\bibnamefont
  {Kafatos}}}\ (\bibinfo  {publisher} {Springer},\ \bibinfo {year} {1989})\
  pp.\ \bibinfo {pages} {69--72}\BibitemShut {NoStop}%
\bibitem [{\citenamefont {Proietti}\ \emph {et~al.}(2021)\citenamefont
  {Proietti}, \citenamefont {Ho}, \citenamefont {Grasselli}, \citenamefont
  {Barrow}, \citenamefont {Malik},\ and\ \citenamefont
  {Fedrizzi}}]{proietti2021experimental}%
  \BibitemOpen
  \bibfield  {author} {\bibinfo {author} {\bibfnamefont {M.}~\bibnamefont
  {Proietti}}, \bibinfo {author} {\bibfnamefont {J.}~\bibnamefont {Ho}},
  \bibinfo {author} {\bibfnamefont {F.}~\bibnamefont {Grasselli}}, \bibinfo
  {author} {\bibfnamefont {P.}~\bibnamefont {Barrow}}, \bibinfo {author}
  {\bibfnamefont {M.}~\bibnamefont {Malik}},\ and\ \bibinfo {author}
  {\bibfnamefont {A.}~\bibnamefont {Fedrizzi}},\ }\href
  {https://doi.org/10.1126/sciadv.abe0395} {\bibfield  {journal} {\bibinfo
  {journal} {Sci. Adv.}\ }\textbf {\bibinfo {volume} {7}},\ \bibinfo {pages}
  {eabe0395} (\bibinfo {year} {2021})}\BibitemShut {NoStop}%
\bibitem [{\citenamefont {Epping}\ \emph {et~al.}(2017)\citenamefont {Epping},
  \citenamefont {Kampermann}, \citenamefont {macchiavello},\ and\ \citenamefont
  {Bruß}}]{epping2017multi}%
  \BibitemOpen
  \bibfield  {author} {\bibinfo {author} {\bibfnamefont {M.}~\bibnamefont
  {Epping}}, \bibinfo {author} {\bibfnamefont {H.}~\bibnamefont {Kampermann}},
  \bibinfo {author} {\bibfnamefont {C.}~\bibnamefont {macchiavello}},\ and\
  \bibinfo {author} {\bibfnamefont {D.}~\bibnamefont {Bruß}},\ }\href
  {https://doi.org/10.1088/1367-2630/aa8487} {\bibfield  {journal} {\bibinfo
  {journal} {New J. Phys.}\ }\textbf {\bibinfo {volume} {19}},\ \bibinfo
  {pages} {093012} (\bibinfo {year} {2017})}\BibitemShut {NoStop}%
\bibitem [{\citenamefont {Hillery}\ \emph {et~al.}(1999)\citenamefont
  {Hillery}, \citenamefont {Bu\ifmmode~\check{z}\else \v{z}\fi{}ek},\ and\
  \citenamefont {Berthiaume}}]{Hillery1999Qsecret}%
  \BibitemOpen
  \bibfield  {author} {\bibinfo {author} {\bibfnamefont {M.}~\bibnamefont
  {Hillery}}, \bibinfo {author} {\bibfnamefont {V.}~\bibnamefont
  {Bu\ifmmode~\check{z}\else \v{z}\fi{}ek}},\ and\ \bibinfo {author}
  {\bibfnamefont {A.}~\bibnamefont {Berthiaume}},\ }\href
  {https://doi.org/10.1103/PhysRevA.59.1829} {\bibfield  {journal} {\bibinfo
  {journal} {Phys. Rev. A}\ }\textbf {\bibinfo {volume} {59}},\ \bibinfo
  {pages} {1829} (\bibinfo {year} {1999})}\BibitemShut {NoStop}%
\bibitem [{\citenamefont {Hahn}\ \emph {et~al.}(2020)\citenamefont {Hahn},
  \citenamefont {de~Jong},\ and\ \citenamefont {Pappa}}]{ACKA}%
  \BibitemOpen
  \bibfield  {author} {\bibinfo {author} {\bibfnamefont {F.}~\bibnamefont
  {Hahn}}, \bibinfo {author} {\bibfnamefont {J.}~\bibnamefont {de~Jong}},\ and\
  \bibinfo {author} {\bibfnamefont {A.}~\bibnamefont {Pappa}},\ }\href
  {https://doi.org/10.1103/PRXQuantum.1.020325} {\bibfield  {journal} {\bibinfo
   {journal} {PRX Quantum}\ }\textbf {\bibinfo {volume} {1}},\ \bibinfo {pages}
  {020325} (\bibinfo {year} {2020})}\BibitemShut {NoStop}%
\bibitem [{\citenamefont {R\"uckle}\ \emph {et~al.}(2023)\citenamefont
  {R\"uckle}, \citenamefont {Budde}, \citenamefont {de~Jong}, \citenamefont
  {Hahn}, \citenamefont {Pappa},\ and\ \citenamefont {Barz}}]{expACKA}%
  \BibitemOpen
  \bibfield  {author} {\bibinfo {author} {\bibfnamefont {L.}~\bibnamefont
  {R\"uckle}}, \bibinfo {author} {\bibfnamefont {J.}~\bibnamefont {Budde}},
  \bibinfo {author} {\bibfnamefont {J.}~\bibnamefont {de~Jong}}, \bibinfo
  {author} {\bibfnamefont {F.}~\bibnamefont {Hahn}}, \bibinfo {author}
  {\bibfnamefont {A.}~\bibnamefont {Pappa}},\ and\ \bibinfo {author}
  {\bibfnamefont {S.}~\bibnamefont {Barz}},\ }\href
  {https://doi.org/10.1103/PhysRevResearch.5.033222} {\bibfield  {journal}
  {\bibinfo  {journal} {Phys. Rev. Res.}\ }\textbf {\bibinfo {volume} {5}},\
  \bibinfo {pages} {033222} (\bibinfo {year} {2023})}\BibitemShut {NoStop}%
\bibitem [{\citenamefont {Raussendorf}\ and\ \citenamefont
  {Briegel}(2001)}]{Raussendorf2001}%
  \BibitemOpen
  \bibfield  {author} {\bibinfo {author} {\bibfnamefont {R.}~\bibnamefont
  {Raussendorf}}\ and\ \bibinfo {author} {\bibfnamefont {H.~J.}\ \bibnamefont
  {Briegel}},\ }\href {https://doi.org/10.1103/physrevlett.86.5188} {\bibfield
  {journal} {\bibinfo  {journal} {Phys. Rev. Lett.}\ }\textbf {\bibinfo
  {volume} {86}},\ \bibinfo {pages} {5188} (\bibinfo {year}
  {2001})}\BibitemShut {NoStop}%
\bibitem [{\citenamefont {Nielsen}(2006)}]{Nielsen_2006_cluster_QC}%
  \BibitemOpen
  \bibfield  {author} {\bibinfo {author} {\bibfnamefont {M.~A.}\ \bibnamefont
  {Nielsen}},\ }\href {https://doi.org/10.1016/s0034-4877(06)80014-5}
  {\bibfield  {journal} {\bibinfo  {journal} {Rep. Math. Phys.}\ }\textbf
  {\bibinfo {volume} {57}},\ \bibinfo {pages} {147} (\bibinfo {year}
  {2006})}\BibitemShut {NoStop}%
\bibitem [{\citenamefont {Danielsen}(2005)}]{danielsen2005database12qubits}%
  \BibitemOpen
  \bibfield  {author} {\bibinfo {author} {\bibfnamefont {L.~E.}\ \bibnamefont
  {Danielsen}},\ }\href@noop {} {\bibinfo {title} {On self-dual quantum codes,
  graphs, and {B}oolean functions}} (\bibinfo {year} {2005}),\ \Eprint
  {https://arxiv.org/abs/quant-ph/0503236} {arXiv:quant-ph/0503236}
  \BibitemShut {NoStop}%
\bibitem [{\citenamefont {Gittsovich}\ \emph {et~al.}(2010)\citenamefont
  {Gittsovich}, \citenamefont {Hyllus},\ and\ \citenamefont
  {Gühne}}]{Gittsovich2010Multiparticle}%
  \BibitemOpen
  \bibfield  {author} {\bibinfo {author} {\bibfnamefont {O.}~\bibnamefont
  {Gittsovich}}, \bibinfo {author} {\bibfnamefont {P.}~\bibnamefont {Hyllus}},\
  and\ \bibinfo {author} {\bibfnamefont {O.}~\bibnamefont {Gühne}},\ }\href
  {https://doi.org/10.1103/physreva.82.032306} {\bibfield  {journal} {\bibinfo
  {journal} {Phys. Rev. A}\ }\textbf {\bibinfo {volume} {82}},\ \bibinfo
  {pages} {032306} (\bibinfo {year} {2010})}\BibitemShut {NoStop}%
\bibitem [{\citenamefont {Cabello}\ \emph {et~al.}(2011)\citenamefont
  {Cabello}, \citenamefont {Danielsen}, \citenamefont {López-Tarrida},\ and\
  \citenamefont {Portillo}}]{cabelloOptimalPreparationGraph2011a}%
  \BibitemOpen
  \bibfield  {author} {\bibinfo {author} {\bibfnamefont {A.}~\bibnamefont
  {Cabello}}, \bibinfo {author} {\bibfnamefont {L.~E.}\ \bibnamefont
  {Danielsen}}, \bibinfo {author} {\bibfnamefont {A.~J.}\ \bibnamefont
  {López-Tarrida}},\ and\ \bibinfo {author} {\bibfnamefont {J.~R.}\
  \bibnamefont {Portillo}},\ }\href
  {https://doi.org/10.1103/PhysRevA.83.042314} {\bibfield  {journal} {\bibinfo
  {journal} {Physical Review A}\ }\textbf {\bibinfo {volume} {83}},\ \bibinfo
  {pages} {042314} (\bibinfo {year} {2011})},\ \bibinfo {note} {publisher:
  American Physical Society}\BibitemShut {NoStop}%
\bibitem [{\citenamefont {Claudet}\ and\ \citenamefont
  {Perdrix}(2025)}]{claudet2024covering}%
  \BibitemOpen
  \bibfield  {author} {\bibinfo {author} {\bibfnamefont {N.}~\bibnamefont
  {Claudet}}\ and\ \bibinfo {author} {\bibfnamefont {S.}~\bibnamefont
  {Perdrix}},\ }in\ \href
  {https://doi.org/https://doi.org/10.1007/978-3-031-75409-8_10} {\emph
  {\bibinfo {booktitle} {Graph-Theoretic Concepts in Computer Science}}},\
  \bibinfo {editor} {edited by\ \bibinfo {editor} {\bibfnamefont
  {D.}~\bibnamefont {Kr{\'a}{\v{l}}}}\ and\ \bibinfo {editor} {\bibfnamefont
  {M.}~\bibnamefont {Milani{\v{c}}}}}\ (\bibinfo  {publisher} {Springer},\
  \bibinfo {year} {2025})\ pp.\ \bibinfo {pages} {136--150},\ \bibinfo {note}
  {arXiv:2402.10678}\BibitemShut {NoStop}%
\bibitem [{\citenamefont {Zhang}(2023)}]{zhang2023bell}%
  \BibitemOpen
  \bibfield  {author} {\bibinfo {author} {\bibfnamefont {D.}~\bibnamefont
  {Zhang}},\ }\href@noop {} {\bibinfo {title} {Bell pair extraction using graph
  foliage techniques}} (\bibinfo {year} {2023}),\ \Eprint
  {https://arxiv.org/abs/2311.16188} {arXiv:2311.16188} \BibitemShut {NoStop}%
\bibitem [{\citenamefont {Kraus}(2010{\natexlab{b}})}]{Kraus2010_PRA}%
  \BibitemOpen
  \bibfield  {author} {\bibinfo {author} {\bibfnamefont {B.}~\bibnamefont
  {Kraus}},\ }\href {https://doi.org/10.1103/PhysRevA.82.032121} {\bibfield
  {journal} {\bibinfo  {journal} {Phys. Rev. A}\ }\textbf {\bibinfo {volume}
  {82}},\ \bibinfo {pages} {032121} (\bibinfo {year}
  {2010}{\natexlab{b}})}\BibitemShut {NoStop}%
\bibitem [{\citenamefont {Wyderka}\ and\ \citenamefont
  {Gühne}(2020)}]{Wyderka_2020}%
  \BibitemOpen
  \bibfield  {author} {\bibinfo {author} {\bibfnamefont {N.}~\bibnamefont
  {Wyderka}}\ and\ \bibinfo {author} {\bibfnamefont {O.}~\bibnamefont
  {Gühne}},\ }\href {https://doi.org/10.1088/1751-8121/ab7f0a} {\bibfield
  {journal} {\bibinfo  {journal} {J. Phys. A Math. Theor.}\ }\textbf {\bibinfo
  {volume} {53}},\ \bibinfo {pages} {345302} (\bibinfo {year}
  {2020})}\BibitemShut {NoStop}%
\bibitem [{\citenamefont {Nguyen}\ and\ \citenamefont
  {il~Oum}(2020)}]{NGUYEN2020103183}%
  \BibitemOpen
  \bibfield  {author} {\bibinfo {author} {\bibfnamefont {H.-T.}\ \bibnamefont
  {Nguyen}}\ and\ \bibinfo {author} {\bibfnamefont {S.}~\bibnamefont
  {il~Oum}},\ }\href
  {https://doi.org/https://doi.org/10.1016/j.ejc.2020.103183} {\bibfield
  {journal} {\bibinfo  {journal} {Eur. J. Comb.}\ }\textbf {\bibinfo {volume}
  {90}},\ \bibinfo {pages} {103183} (\bibinfo {year} {2020})}\BibitemShut
  {NoStop}%
\bibitem [{\citenamefont {H{\o}yer}\ \emph {et~al.}(2006)\citenamefont
  {H{\o}yer}, \citenamefont {Mhalla},\ and\ \citenamefont
  {Perdrix}}]{10.1007/11940128_64}%
  \BibitemOpen
  \bibfield  {author} {\bibinfo {author} {\bibfnamefont {P.}~\bibnamefont
  {H{\o}yer}}, \bibinfo {author} {\bibfnamefont {M.}~\bibnamefont {Mhalla}},\
  and\ \bibinfo {author} {\bibfnamefont {S.}~\bibnamefont {Perdrix}},\ }in\
  \href {https://doi.org/10.1007/978-3-642-25591-5} {\emph {\bibinfo
  {booktitle} {Algorithms and Computation}}},\ \bibinfo {editor} {edited by\
  \bibinfo {editor} {\bibfnamefont {T.}~\bibnamefont {Asano}}}\ (\bibinfo
  {publisher} {Springer},\ \bibinfo {year} {2006})\ pp.\ \bibinfo {pages}
  {638--649}\BibitemShut {NoStop}%
\bibitem [{\citenamefont {Hahn}\ \emph {et~al.}(2019)\citenamefont {Hahn},
  \citenamefont {Pappa},\ and\ \citenamefont
  {Eisert}}]{Hahn_2019_network_routing}%
  \BibitemOpen
  \bibfield  {author} {\bibinfo {author} {\bibfnamefont {F.}~\bibnamefont
  {Hahn}}, \bibinfo {author} {\bibfnamefont {A.}~\bibnamefont {Pappa}},\ and\
  \bibinfo {author} {\bibfnamefont {J.}~\bibnamefont {Eisert}},\ }\href
  {https://doi.org/10.1038/s41534-019-0191-6} {\bibfield  {journal} {\bibinfo
  {journal} {npj Quantum Inf.}\ }\textbf {\bibinfo {volume} {5}},\ \bibinfo
  {pages} {76} (\bibinfo {year} {2019})}\BibitemShut {NoStop}%
\bibitem [{\citenamefont {Mannalath}\ and\ \citenamefont
  {Pathak}(2023)}]{Mannalath2023multiparty}%
  \BibitemOpen
  \bibfield  {author} {\bibinfo {author} {\bibfnamefont {V.}~\bibnamefont
  {Mannalath}}\ and\ \bibinfo {author} {\bibfnamefont {A.}~\bibnamefont
  {Pathak}},\ }\href {https://doi.org/10.1103/PhysRevA.108.062614} {\bibfield
  {journal} {\bibinfo  {journal} {Phys. Rev. A}\ }\textbf {\bibinfo {volume}
  {108}},\ \bibinfo {pages} {062614} (\bibinfo {year} {2023})}\BibitemShut
  {NoStop}%
\bibitem [{\citenamefont {de~Jong}\ \emph {et~al.}(2024)\citenamefont
  {de~Jong}, \citenamefont {Hahn}, \citenamefont {Tcholtchev}, \citenamefont
  {Hauswirth},\ and\ \citenamefont {Pappa}}]{dejong2023extracting}%
  \BibitemOpen
  \bibfield  {author} {\bibinfo {author} {\bibfnamefont {J.}~\bibnamefont
  {de~Jong}}, \bibinfo {author} {\bibfnamefont {F.}~\bibnamefont {Hahn}},
  \bibinfo {author} {\bibfnamefont {N.}~\bibnamefont {Tcholtchev}}, \bibinfo
  {author} {\bibfnamefont {M.}~\bibnamefont {Hauswirth}},\ and\ \bibinfo
  {author} {\bibfnamefont {A.}~\bibnamefont {Pappa}},\ }\href
  {https://doi.org/10.1103/PhysRevResearch.6.013330} {\bibfield  {journal}
  {\bibinfo  {journal} {Phys. Rev. Res.}\ }\textbf {\bibinfo {volume} {6}},\
  \bibinfo {pages} {013330} (\bibinfo {year} {2024})}\BibitemShut {NoStop}%
\bibitem [{\citenamefont {Brand}\ \emph {et~al.}(2023)\citenamefont {Brand},
  \citenamefont {Coopmans},\ and\ \citenamefont {Laarman}}]{brand2023quantum}%
  \BibitemOpen
  \bibfield  {author} {\bibinfo {author} {\bibfnamefont {S.}~\bibnamefont
  {Brand}}, \bibinfo {author} {\bibfnamefont {T.}~\bibnamefont {Coopmans}},\
  and\ \bibinfo {author} {\bibfnamefont {A.}~\bibnamefont {Laarman}},\
  }\href@noop {} {\bibinfo {title} {Quantum graph-state synthesis with {SAT}}}
  (\bibinfo {year} {2023}),\ \Eprint {https://arxiv.org/abs/2309.03593}
  {arXiv:2309.03593} \BibitemShut {NoStop}%
\bibitem [{\citenamefont {Szymański}\ \emph {et~al.}(2024)\citenamefont
  {Szymański}, \citenamefont {Vandré},\ and\ \citenamefont
  {Gühne}}]{szymanski2024useful}%
  \BibitemOpen
  \bibfield  {author} {\bibinfo {author} {\bibfnamefont {K.}~\bibnamefont
  {Szymański}}, \bibinfo {author} {\bibfnamefont {L.}~\bibnamefont
  {Vandré}},\ and\ \bibinfo {author} {\bibfnamefont {O.}~\bibnamefont
  {Gühne}},\ }\href@noop {} {\bibinfo {title} {Useful entanglement can be
  extracted from noisy graph states}} (\bibinfo {year} {2024}),\ \Eprint
  {https://arxiv.org/abs/2402.00937} {arXiv:2402.00937} \BibitemShut {NoStop}%
\end{thebibliography}%

\end{document}